\documentclass[lettersize,journal]{IEEEtran}
\usepackage{amsmath,amssymb,amsfonts}
\usepackage{amsthm}
\usepackage{array}
\usepackage{algorithmic}
\usepackage{algorithm}
\usepackage{balance}
\usepackage{caption}
\usepackage{graphicx}
\usepackage{epstopdf}
\usepackage{float}
\usepackage{subcaption}
\usepackage{textcomp}
\usepackage{xcolor}
\usepackage{stfloats}
\usepackage{verbatim}
\usepackage{cite}
\usepackage{bm}
\usepackage{lettrine}
\graphicspath{ {figs/} }

\IEEEaftertitletext{\vspace{-8pt}}
\newtheorem{theorem}{Theorem}
\newtheorem{lemma}{Lemma}

\theoremstyle{definition}

\begin{document}

\title{Integrated Sensing, Communication, and Control for UAV-Assisted Mobile Target Tracking}

\author{\IEEEauthorblockN{Zhiyu Chen, \IEEEmembership{Student Member, IEEE}, Ming-Min Zhao, \IEEEmembership{Senior Member, IEEE}, Songfu Cai, \IEEEmembership{Member, IEEE}, Ming Lei, \IEEEmembership{Member, IEEE}, and Min-Jian Zhao, \IEEEmembership{Member, IEEE}
		\thanks{Z. Chen, M. M. Zhao, S. Cai, M. Lei, and M. J. Zhao are with the College of Information Science and Electronic Engineering, Zhejiang University, Hangzhou 310027, China, and also with Zhejiang Provincial Key Laboratory of Multi-Modal Communication Networks and Intelligent Information Processing, Hangzhou 310027, China (e-mail: \{12431083, zmmblack, sfcai, lm1029, mjzhao\}@zju.edu.cn).}
	}
}

\maketitle

\begin{abstract}
	Unmanned aerial vehicles (UAVs) are increasingly deployed in mission-critical applications such as target tracking, where they must simultaneously sense dynamic environments, ensure reliable communication, and achieve precise control. A key challenge here is to jointly guarantee tracking accuracy, communication reliability, and control stability within a unified framework. To address this issue, we propose an integrated sensing, communication, and control (ISCC) framework for UAV-assisted target tracking, where the considered tracking system is modeled as a discrete-time linear control process, with the objective of driving the deviation between the UAV and target states toward zero. We formulate a stochastic model predictive control (MPC) optimization problem for joint control and beamforming design, which is highly non-convex and intractable in its original form. To overcome this difficulty, the target state is first estimated using an extended Kalman filter (EKF). Then, by deriving the closed-form optimal beamforming solution under a given control input, the original problem is equivalently reformulated into a tractable control-oriented form. Finally, we convexify the remaining non-convex constraints via a relaxation-based convex approximation, yielding a computationally tractable convex optimization problem that admits efficient global solution. Numerical results show that the proposed ISCC framework achieves tracking accuracy comparable to a non-causal benchmark while maintaining stable communication, and it significantly outperforms the conventional control and tracking method.
\end{abstract}


\section{Introduction}
\IEEEPARstart{I}{ntegrated} sensing and communication (ISAC) has emerged as a key technology for the sixth-generation (6G) wireless networks, designed to address the growing demand for more efficient spectrum utilization, reduced hardware costs, and improved performance\cite{I1,I2,I3,I4}. Traditionally, sensing and communication systems are developed separately, with each operating on distinct frequency bands. However, with the rapid growth of Internet of Things (IoT) devices, spectrum scarcity has become a significant challenge for both radar and communication\cite{I5}. Moreover, the forthcoming 6G networks are expected to support a wide range of environment-aware applications, such as autonomous vehicles, smart cities, and mobile virtual reality, which impose high demands on both communication quality and sensing accuracy\cite{I1},\cite{I6}. In response to these challenges, ISAC offers a promising solution by enabling the simultaneous operation of sensing and communication on shared radio spectrum and hardware platforms. By effectively sharing resources across spatial, temporal, and frequency domains, the ISAC technology not only significantly reduces hardware complexity but also alleviates the strain of spectrum scarcity\cite{I5},\cite{I7,I8}.

Extensive efforts have been devoted to multi-antenna ISAC systems, with a primary focus on transmit beamforming and/or waveform design under coupled sensing and communication metrics \cite{I9,I10,I11,I12,I13,I14,I15}. For example, \cite{I9} optimizes transmit beamforming to improve radar beampattern performance while satisfying communication quality-of-service requirements such as the signal-to-interference-plus-noise ratio (SINR). Beyond reusing information-bearing signals for sensing, \cite{I10} and \cite{I11} further propose dedicated sensing signal transmission to provide additional degrees of freedom (DoFs) and improve the joint performance. Despite these advances, most existing results are developed for terrestrial ISAC deployments, where sensing may suffer inherent limitations in practical propagation environments. Specifically, target detection and parameter estimation often rely on line-of-sight (LoS) links between sensing transceivers (e.g., ISAC base stations (BSs) or access points (APs)) and targets, whereas surrounding obstacles and scatterers may block LoS links and introduce non-LoS (NLoS) paths and clutter, thereby degrading sensing reliability. Moreover, high-precision long-range sensing typically requires sufficient transmit power; yet terrestrial BSs/APs are often power-limited. Consequently, for distant or obstructed targets, the echoed signals may experience severe round-trip path loss, resulting in significant sensing degradation or even sensing failure. Therefore, terrestrial ISAC BSs/APs can usually provide effective sensing and communication only within a limited range and under favorable propagation conditions.

Motivated by the rapid development of unmanned aerial vehicle (UAV)-enabled wireless communications \cite{U1}, UAVs are anticipated to serve as aerial ISAC platforms to alleviate the above limitations \cite{b23}. Compared with terrestrial BSs/APs, UAVs operating at elevated altitudes can establish air-to-ground (A2G) links with strong LoS components with high probability \cite{U3,U4,U5,U6,U7}, which benefits both reliable communication and sensing signal acquisition. In scenarios where terrestrial infrastructure is unavailable or temporarily overloaded (e.g., post-disaster recovery and temporary outdoor hot spots), UAVs can be rapidly deployed as aerial BSs to provide continuous communication services for ground users \cite{U8}. More importantly, UAV mobility enables flexible positioning and trajectory design \cite{U9,U10}, allowing UAVs to approach regions of interest and adapt the propagation geometry with respect to both users and targets. This capability improves LoS availability and mitigates severe round-trip path loss of echoed signals, thereby enhancing sensing reliability and extending the service coverage of ISAC systems.

Recently, UAV-enabled ISAC networks have attracted growing attention by leveraging UAV mobility and favorable A2G LoS propagation to enhance both sensing reliability and communication coverage \cite{b3,U2,b21,b24,b4,b9,b8,C1,C2}. For instance, the works \cite{b3,U2,b21} mainly consider stationary sensing targets and investigate joint UAV trajectory and transmit design under sensing-quality constraints, including periodic ISAC mechanisms that optimize sensing instants and flight trajectories \cite{b3}, multi-user downlink with area sensing via joint positioning and beamforming optimization \cite{U2}, and monostatic UAV-ISAC localization relying on weighted trajectory design to characterize the sensing-communication trade-off \cite{b21}. Furthermore, several works extend UAV-ISAC to dynamic scenarios by explicitly accounting for the continuous motions of both the UAV and the tracked object, such as secure communication with UAV-enabled tracking \cite{b8}, multi-UAV tracking that balances data rate and localization accuracy \cite{b4}, space-air-ground integrated network (SAGIN)-oriented UAV-ISAC tracking with robust coverage and beamforming \cite{b9}, and trajectory optimization that minimizes predicted estimation bounds for both the relative position and the relative velocity \cite{b24}. Despite these representative advances, most existing UAV-ISAC designs still treat UAV motion primarily as a kinematic planning variable and optimize trajectories at the geometric level, without explicitly incorporating the underlying closed-loop control system and its control inputs that govern real-time maneuvering and tracking behavior \cite{b3,U2,b21,b24,b4,b9,b8}. Although several recent studies begin to introduce control variables into UAV motion design \cite{C1,C2}, they often rely on idealized assumptions such as the target locations being known a priori, thereby bypassing a central challenge in practical target tracking where the target state is only partially observed and must be inferred online. These limitations motivate an integrated sensing, communication, and control (ISCC) design that explicitly models UAV-target tracking-error state evolution, couples control actions with sensing/communication constraints, and yields implementable real-time policies in the presence of tracking uncertainty.

Based on the above discussions, this paper investigates the UAV-assisted target tracking system shown in Fig.~\ref{fig1}, where an ISAC-enabled UAV simultaneously tracks a moving target and serves a ground user (GU). To enable accurate tracking while providing reliable communication, we develop an ISCC framework and investigate the real-time joint design of UAV maneuvering and transmit beamforming under practical sensing and communication requirements. This results in a stochastic model predictive control (MPC) formulation that couples control actions, state uncertainty, and beamforming decisions, making the resulting optimization highly non-convex. By leveraging the underlying problem structure, we develop a tractable reformulation together with a relaxation-based convex approximation, which makes the proposed ISCC design amenable to online implementation. The main contributions are summarized as follows.
\begin{itemize}
	\item
	We develop an ISCC framework for UAV-assisted target tracking, where an ISAC-enabled UAV simultaneously tracks a moving target and serves a GU. The tracking task is explicitly cast as a discrete-time linear control process by defining the UAV-target state deviation as the system state and the UAV maneuvering command as the control input. Under this ISCC modeling, we formulate an ISCC-based stochastic MPC problem that jointly optimizes the control inputs and transmit beamforming over a receding horizon to drive the tracking error to zero while satisfying sensing and communication requirements, which is highly non-convex due to the coupling among control actions, state uncertainty, and beamforming constraints.
	
	\item
	We propose a tractable online solution approach for the above ISCC-based stochastic MPC problem. Specifically, an extended Kalman filter (EKF) is employed to estimate and predict the target state required by MPC. Then, by deriving the closed-form optimal beamforming solution for any given control input, we obtain an equivalent control-oriented reformulation that removes the beamforming variables from the coupled joint design. The remaining non-convex constraints are then handled via a convex lower-bound relaxation technique, yielding a convex program that can be efficiently solved.
	
	\item
	Numerical simulations validate the effectiveness of the proposed ISCC framework. The results show that the proposed design achieves tracking performance close to a non-causal MPC benchmark (with access to future target evolution) while maintaining a stable communication service at the required rate threshold, and it consistently outperforms a conventional control and tracking baseline under practical constraints.
\end{itemize}

The remainder of this paper is organized as follows. Section~II presents the system model, including the communication model, sensing and measurement model, and the control-system model. Section~III formulates the stochastic MPC optimization problem for joint control and beamforming design. Section~IV develops the proposed ISCC framework and the corresponding solution approach. Section~V provides numerical results. Finally, Section~VI concludes the paper.

\textit{Notations}: Scalars, vectors, and matrices are denoted by lower-case letters, bold-face lower-case letters, and bold-face upper-case letters, respectively. $(\cdot)^{T}$ and $(\cdot)^{H}$ denote the transpose and conjugate transpose, respectively. $\mathbb{R}^{x\times y}$ and $\mathbb{C}^{x\times y}$ denote the spaces of real and complex matrices of size $x\times y$, respectively. $\|\bm{\alpha}\|$ denotes the Euclidean norm of a complex vector $\bm{\alpha}$, and $|z|$ denotes the magnitude of a complex scalar $z$. $\mathbb{E}\{\cdot\}$ denotes statistical expectation. The symbol $\triangleq$ indicates a definition. $\Re(\cdot)$ and $\Im(\cdot)$ denote the real and imaginary parts of a complex number, respectively. $\mathrm{diag}(\bm{\alpha})$ denotes the diagonal matrix whose diagonal entries are given by the elements of $\bm{\alpha}$. The Hadamard product and Kronecker product are denoted by $\odot$ and $\otimes$, respectively. $\mathbf{I}$ denotes an identity matrix with an appropriate dimension, $\bm{0}$ denotes an all-zero vector or matrix with an appropriate dimension, and $\bold{1}_{n}$ denotes an $n$-dimensional all-one column vector.

\section{System Model}\label{sec:II}
As shown in Fig. \ref{fig1}, we consider a UAV-assisted tracking system, where an ISAC UAV simultaneously tracks a moving target and communicates with a single-antenna GU within a predetermined flight duration $T$. The flight period $T$ is discretized into $N$ sufficiently small time slots using a discrete path planning approach \cite{b5}, with the duration of each slot given by $\delta_t = T/N$. For clarity, let $n\in\mathcal{N} \triangleq \{ 1,\cdots,N\}$ denote the time slot index. Within each time slot, the UAV is assumed to maintain a fixed flight altitude $H$, and the motion states (i.e., the horizontal position and the velocity) of both the target and the UAV remain constant. Specifically, in the $n$-th time slot, their states are given by
\begin{equation}\small
	\bold{s}^{U}[n]=\begin{bmatrix}
		\bold{p}^{U}[n] \\
		\bold{v}^{U}[n] 
	\end{bmatrix},\ 
	\bold{s}^{t}[n]=\begin{bmatrix}
		\bold{p}^{t}[n] \\
		\bold{v}^{t}[n] 
	\end{bmatrix},\
\end{equation} 
where $\bold{p}^{U}[n] \triangleq [p_x^{U}[n], p_y^{U}[n]]^T$ and $\bold{p}^{t}[n] \triangleq [p_x^{t}[n], p_y^{t}[n]]^T$ denote the horizontal positions of the UAV and the target, respectively, while $\bold{v}^{U}[n] \triangleq [v_x^{U}[n], v_y^{U}[n]]^T$ and $\bold{v}^{t}[n] \triangleq [v_x^{t}[n], v_y^{t}[n]]^T$ represent their corresponding velocity vectors. In addition, the UAV employs a uniform planar array (UPA) with half-wavelength spacing along both the $x$- and $y$-axes. The transmit and receive antenna numbers are $M_t = M_x^t \times M_y^t$ and $M_r = M_x^r\times M_y^r$, where $M_x^{t/r}$ and $M_y^{t/r}$ represent the numbers of elements along the $x$- and $y$-axes, respectively.

\begin{figure}[t]
	\centering
	\includegraphics[width=0.95\linewidth]{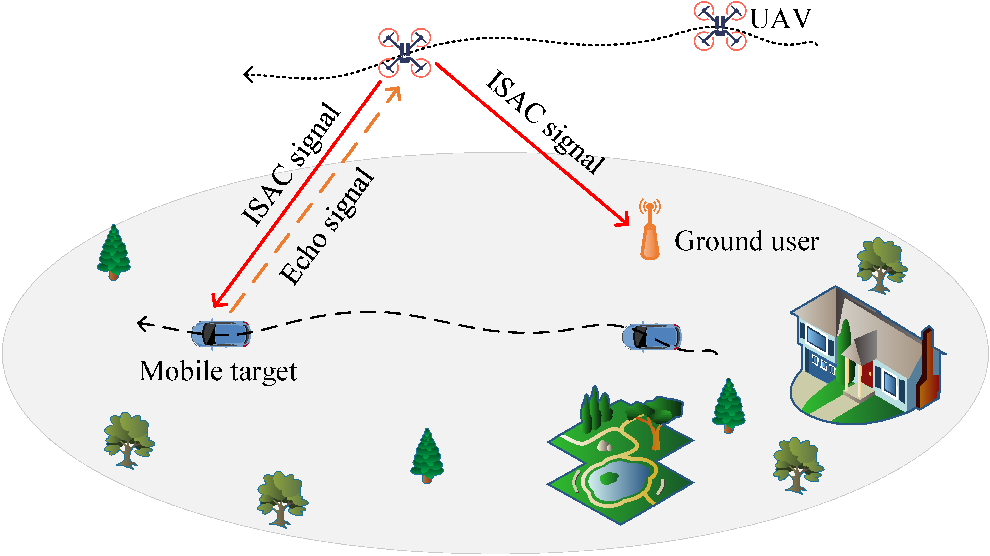}
	\caption{UAV-assisted mobile target tracking system.}
	\label{fig1}
	\vspace{0pt}
\end{figure}

\subsection{Communication Model}
The communication link between the UAV and GU is assumed to be dominated by the LoS component\cite{b3}. Hence, the channel state information (CSI) from the UAV to the GU in the $n$-th time slot can be molded as
\begin{equation}\label{E-8} \small
	\bold{h}^H(\bold{p}^{U}[n],\bold{p}^{u})=\sqrt{\beta_0d(\bold{p}^{U}[n],\bold{p}^{u})^{-2}}\ \bold{a}^H(\bold{p}^{U}[n],\bold{p}^{u}),
\end{equation}
where $\bold{p}^{u}\triangleq[p_x^{u}, p_y^{u}]^T$ denotes the position of the GU, and $\beta_0$ represents the channel power at the reference distance of 1 m. Moreover, $d(\bold{p}^{U}[n],\bold{p}^{u})=\sqrt{(\Vert\bold{p}^{U}[n]-\bold{p}^{u}\Vert^2+H^2)}$ denotes the distance between the UAV and GU at the $n$-th time slot, and $\bold{a}(\bold{p}^{U}[n],\bold{p}^{u})$ is the transmit steering vector from the UAV to the GU, as defined in (\ref{Ts}) at the top of the next page,
\begin{figure*}[ht]  
	\centering 
	\vspace*{0pt} 
	\begin{equation}\label{Ts} \small
		\bold{a}(\bold{p}^{U}[n],\bold{p}^{u})=[1,\dots,e^{j\pi(M_x^t-1) \Phi(\bold{p}^{U}[n],\bold{p}^{u})}]^T \otimes [1,\dots,e^{j\pi(M_y^t-1)\Omega(\bold{p}^{U}[n],\bold{p}^{u})}]^T
	\end{equation}
	\hrulefill 
\end{figure*}
where $\Phi(\bold{p}^{U}[n],\bold{p}^{u})=\frac{p_x^U[n]-p_x^u}{d(\bold{p}^{U}[n],\bold{p}^{u})}$ and $\Omega(\bold{p}^{U}[n],\bold{p}^{u})=\frac{p_y^U[n]-p_y^u}{d(\bold{p}^{U}[n],\bold{p}^{u})}$. 

Without loss of generality, let $s_n(t) \sim \mathcal{CN}(0,1)$ denote the information-bearing signal transmitted to the GU at (continuous) time $t$ within the $n$-th time slot. Then, the transmitted ISAC signal can be denoted by $\bold{x}_n(t)=\bold{w}_ns_n(t)$, where $\bold{w}_n\in \mathbb{C}^{M_t\times 1}$ is the transmit beamforming vector used in the $n$-th time slot. Consequently, the received signal at the GU is
\begin{equation}\label{E-5} \small
	y_n(t) = \bold{h}^H(\bold{p}^{U}[n],\bold{p}^{u})\bold{x}_n(t) + n_c(t).
\end{equation}
where $n_c(t) \sim \mathcal{CN}(0,\sigma_c^2)$ is a circularly symmetric complex Gaussian (CSCG) noise at the GU receiver. The achievable rate (in bps/Hz) at the $n$-th time slot is then given by
\begin{equation}\label{E-6} \small
	R_n=\text{log}_2(1+|\bold{h}^H(\bold{p}^{U}[n],\bold{p}^{u})\bold{w}_n|^2/\sigma_c^2).
\end{equation}

\subsection{Sensing and Measurement Model}
The communication signals $\{s_n(t)\}$ can also be utilized for sensing by leveraging their reflections off the target to estimate its parameters. At the $n$-th time slot, the received echo is
\begin{equation}\label{E-7}\small 
	\bold{r}_n(t) = \beta_n e^{j2\pi \mu_n t} \bold{b}_n \bold{a}^H_n \bold{w}_n s_n(t-\tau_n)+ \bold{n}_r(t),
\end{equation}
where $\bold{a}_n \triangleq \bold{a}(\bold{p}^{U}[n],\bold{p}^{t}[n])$, and $\bold{b}_n \triangleq \bold{b}(\bold{p}^{U}[n],\bold{p}^{t}[n])$ denotes the receive steering vector, which is defined in a similar form as (\ref{Ts}). Moreover, $\bold{n}_r(t)\sim\mathcal{CN}(\mathbf{0},\sigma_r^2 \mathbf{I}_{M_r})$ represents the CSCG noise, $\beta_n = \frac{\sqrt{\beta_r}}{d^2(\bold{p}^{U}[n], \bold{p}^{t}[n])}$ with $\beta_r = \lambda^2 \varepsilon/(64\pi^3)$ denotes the reflection coefficient, where $\lambda$ and $\varepsilon$ denote the wavelength and radar cross section (RCS), respectively. $\mu_n$ and $\tau_n$ correspond to the Doppler shift and the round-trip time delay, respectively. To estimate $\tau_n$ and $\mu_n$, a matched filtering operation is applied to the received signal, which is mathematically expressed as\cite{b6}
\begin{equation}\label{E-18} \small
	\{\hat{\tau}_{n},\hat{\mu}_{n}\}=\arg\max_{\tau,\mu}\left|\int_{0}^{\delta_t}\mathbf{r}_{n}\left(t\right)s_{n}^{*}\left(t-\tau\right)e^{-j2\pi\mu t}dt\right|^{2}.
\end{equation}
Then, by compensating (\ref{E-7}) using $\{\hat{\tau}_{n},\hat{\mu}_{n}\}$, we have
\begin{subequations}\label{E-10}\small
	\begin{align}
		\hat{\bold{r}}_n=& \int_{0}^{\delta_t}\mathbf{r}_{n}\left(t\right)s_{n}^{*}\left(t-\hat{\tau}_{n}\right)e^{-j2\pi\hat{\mu}_{n}t}dt \notag\\
		=&\beta_{n}\bold{b}_n\bold{a}^H_n\bold{w}_n \int_{0}^{\delta_t}s_{n}\left(t-\tau_{n}\right)s_{n}^{*}\left(t-\hat{\tau}_{n}\right)e^{j2\pi(\mu_{n}-\hat{\mu}_{n})t}dt \notag\\
		& +\int_0^{\delta_t}\mathbf{n}_{r}\left(t\right)s_{n}^*\left(t-\hat{\tau}_{n}\right)e^{-j2\pi\hat{\mu}_{n}t}dt \notag\\
		\triangleq& G\beta_{n}\bold{b}_n\bold{a}^H_n\bold{w}_n+\mathbf{\tilde{n}}_r,\tag{8}
	\end{align}
\end{subequations}
where $G\approx N_{\text{sym}}$ is the matched-filtering gain, with $N_{\text{sym}}$ being the number of symbols received during each time slot, and $\tilde{\bold{n}}_r \sim \mathcal{CN}(\mathbf{0}, N_{\text{sym}}\sigma_r^2 \mathbf{I}_{M_r})$ denotes the effective noise vector. To derive a measurement model for the target state $\bold{s}^{t}[n]$, the complex-valued expression in (\ref{E-10}) is reformulated into the following real-valued form:
\begin{equation}\label{E-11} \small
	\bm{\rho}_n \triangleq
	\begin{bmatrix}
		\Re\left\{\hat{\bold{r}}_n\right\} \\
		\Im\left\{\hat{\bold{r}}_n\right\}
	\end{bmatrix}
	=  \begin{bmatrix}
		\Re\left\{G\beta_{n}\bold{b}_n\bold{a}^H_n\bold{w}_n\right\} \\
		\Im\left\{G\beta_{n}\bold{b}_n\bold{a}^H_n\bold{w}_n\right\}
	\end{bmatrix} + \mathbf{n}_{\rho},
\end{equation}
where
\begin{equation} \small
	\mathbf{n}_{\rho} \triangleq
	\begin{bmatrix}
		\Re\left\{\mathbf{\tilde{n}}_r\right\} \\
		\Im\left\{\mathbf{\tilde{n}}_r\right\}
	\end{bmatrix} \sim \mathcal{N}(\bold{0},\frac{N_{\text{sym}}}{2}\sigma_r^2 \mathbf{I}_{2M_r}).
\end{equation}

To facilitate the target state estimation, we further express the sensing observation in terms of the range and Doppler measurements extracted from the matched-filter output. Specifically, the measurement models associated with the target position $\bold p^t[n]$ and velocity $\bold v^t[n]$ are given by
\vspace{-5pt}

\begin{subequations}\label{E-12}\small
	\begin{align} 
		& \hat{d}_n \triangleq \frac{c\hat{\tau}_{n}}{2}=d(\bold{p}^{U}[n],\bold{p}^{t}[n])+n_d, \\
		& \hat{\mu}_{n}=-\frac{2f_c(\bold{p}^{U}[n]-\bold{p}^{t}[n])^T(\bold{v}^{U}[n]-\bold{v}^{t}[n])}{cd(\bold{p}^{U}[n],\bold{p}^{t}[n])}+n_\mu,
	\end{align}
\end{subequations}
where $f_c$ and $c$ represent the carrier frequency and the speed of light, respectively, $n_d\sim\mathcal{N}(0,\sigma_1^2)$ and $n_\mu\sim\mathcal{N}(0,\sigma_2^2)$ denote the measurement noises with $
	\sigma_i^2=\frac{a_i^2\sigma_r^2}{GM_r|\beta_{n}|^2|\bold{a}^H_n\bold{w}_n|^2}, \ i=1,2$,
$a_1$ and $a_2$ are constants related to the system configuration, signal designs as well as the specific signal processing algorithms\cite{b6,b7}. 

Based on (\ref{E-11}) and (\ref{E-12}), the overall mobile target state measurement model can be compactly rewritten as
\begin{equation}\label{M1}\small
	\mathbf m_n=\mathbf f(\mathbf s_t[n])+\mathbf z_n, 
\end{equation}
where $\mathbf m_n \triangleq [\hat d_n,\hat\mu_n,\boldsymbol\rho_n^T]^T$ denotes the state measurement, and $\mathbf z_n \triangleq [n_d,n_\mu,\mathbf n_\rho^T]^T$ represents the measurement noise, which is modeled as a zero-mean Gaussian random vector with covariance matrix $\bold{Q}_m=\text{diag}([\sigma_1^2,\sigma_2^2,\frac{G\sigma_r^2}{2}\bold{1}^T_{2M_r}])$. The nonlinear measurement function $\mathbf f(\cdot)$ is given by
\begin{equation}\small
	\mathbf f(\mathbf s_t[n]) \triangleq
	\begin{bmatrix}
		d(\mathbf p^U[n],\mathbf p^t[n])\\[1mm]
		-\dfrac{2f_c(\bold{p}^{U}[n]-\bold{p}^{t}[n])^T(\bold{v}^{U}[n]-\bold{v}^{t}[n])}{cd(\bold{p}^{U}[n],\bold{p}^{t}[n])}\\[2mm]
		\Re\left\{G\beta_{n}\bold{b}_n\bold{a}^H_n\bold{w}_n\right\}\\
		\Im\left\{G\beta_{n}\bold{b}_n\bold{a}^H_n\bold{w}_n\right\}
	\end{bmatrix}.\notag
\end{equation}

\subsection{Control System Model}
To characterize the motion of both the UAV and target over the considered period $T$, we model their dynamics using the discrete-time kinematic equation\cite{b8}. First, we consider a deterministic UAV control system, whose motion is governed by the following kinematic equations:
\begin{equation} \small
	\left\{\begin{aligned}
		p_x^U[n] &= p_x^U[n-1] + v_x^U[n-1]\delta_t + \tfrac{1}{2}a_x[n-1]\delta_t^2,\\
		p_y^U[n] &= p_y^U[n-1] + v_y^U[n-1]\delta_t + \tfrac{1}{2}a_y[n-1]\delta_t^2,\\
		v_x^U[n] &= v_x^U[n-1] + a_x[n-1]\delta_t,\\
		v_y^U[n] &= v_y^U[n-1] + a_y[n-1]\delta_t,
	\end{aligned}
	\right.
\end{equation}
which can be equivalently expressed in matrix form as
\begin{equation}\label{S1} \small
	\bold{s}^U[n+1] = \bold{A}\bold{s}^U[n] + \bold{B}\bold{u}_n,
\end{equation}
where $\bold{u}_n\triangleq [a_x,a_y]^T \in \mathbb{R}^{2} $ represents the control input, and the constant matrices $\bold{A}$ and $\bold{B}$ are given by
\begin{equation}\label{AB} \small
	\bold{A}=\begin{bmatrix}
		1 & 0 & \delta_t & 0 \\
		0 & 1 & 0 & \delta_t \\
		0 & 0 & 1 & 0 \\
		0 & 0 & 0 & 1
	\end{bmatrix},\ 
	\bold{B}=\begin{bmatrix}
		\tfrac{1}{2}\delta_t^2 & 0 \\
		0 & \tfrac{1}{2}\delta_t^2 \\
		\delta_t & 0 \\
		0 & \delta_t
	\end{bmatrix}.
\end{equation}
Next, we assume that the target moves at an approximately constant velocity within each time slot\cite{b8,b9,b4}, then its motion can be modeled by the following kinematic equations:
\begin{equation}\label{T1} \small
	\left\{\begin{aligned}
		p_x^t[n] &= p_x^t[n-1] + v_x^t[n-1]\delta_t + n_x,\\
		p_y^t[n] &= p_y^t[n-1] + v_y^t[n-1]\delta_t + n_y,\\
		v_x^t[n] &= v_x^t[n-1] + n_{v_x},\\
		v_y^t[n] &= v_y^t[n-1] + n_{v_y},
	\end{aligned}
	\right.
\end{equation}
where $n_x$, $n_y$, $n_{v_x}$, and $n_{v_y}$ denote the corresponding transition noise terms, which follow the distributions $\mathcal{N}(0,\sigma_x^2)$, $\mathcal{N}(0,\sigma_x^2)$, $\mathcal{N}(0,\sigma_x^2)$ and $\mathcal{N}(0,\sigma_{v_y}^2)$, respectively. Similarly, equation (\ref{T1}) can be expressed in matrix form as
\begin{equation}\label{S2} \small
	\bold{s}^t[n+1] = \bold{A}\bold{s}^t[n] + \bold{n}_s,
\end{equation}
where $\bold{n}_s = [n_x, n_y, n_{v_x}, n_{v_y}]^T$ is modeled as a zero-mean Gaussian random vector with covariance matrix $\bold Q_s \triangleq \text{diag}([\sigma_x^2,\sigma_y^2,\sigma_{v_x}^2,\sigma_{v_y}^2])$. Combining \eqref{S1} and \eqref{S2}, the difference between the UAV and target states evolves according to
\begin{equation}\label{S3} \small
	\bold{s}^U[n+1]-\bold{s}^t[n+1] = \bold{A}\big(\bold{s}^U[n]-\bold{s}^t[n]\big)+\bold{B}\bold{u}_n-\bold{n}_s.
\end{equation}
By defining the target state tracking error as $\bold{e}_n\triangleq \bold{s}^U[n]-\bold{s}^t[n]$, (\ref{S3}) can be equivalently expressed in the following compact form:
\begin{equation}\label{S4} \small
	\bold{e}_{n+1} = \bold{A}\bold{e}_n + \bold{B}\bold{u}_n - \bold{n}_s,
\end{equation}
which describes a discrete-time linear control system in which the state tracking error $\bold{e}_n$ serves as the system state.

\section{Problem Formulation}\label{sec:III}
In this paper, we aim to regulate the target state tracking error $\bold{e}_n$ around the zero reference, thereby enabling smooth target tracking by jointly optimizing the control input $\bold{u}_n$ and the transmit beamforming $\bold{w}_n$, under communication, sensing, and transmit power constraints at each time slot. Achieving this requires a carefully designed objective function. 

To construct such an objective, we first recall the classical optimal control method in control theory, i.e., linear quadratic regulator (LQR)\cite{b11}. As a standard paradigm for solving closed-loop optimal control problems in linear dynamical systems, LQR strikes a balance between target state tracking deviation and control effort, and yields an optimal state-feedback control law in closed form. A typical closed-loop feedback control system with an LQR traget tracking controller is illustrated in Fig. \ref{LQR_framework}.
\begin{figure}[!h]
	\centering
	\includegraphics[width=0.95\linewidth]{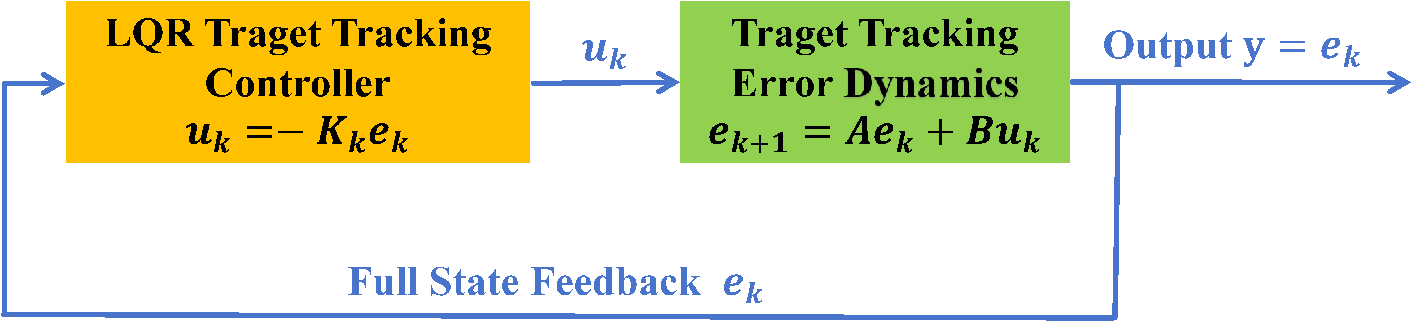}
	\caption{Illustration of a typical closed-loop feedback control system with LQR traget tracking controller.}
	\label{LQR_framework}
	\vspace{0pt}
\end{figure}

In the control system illustrated in Fig. \ref{LQR_framework}, the LQR traget tracking controller defines the following quadratic cost function over the entire mission horizon $[0, N]$:
\begin{equation}\label{LQR-1}\small
	\mathcal{J}(\bold{u}_k)=\sum_{k=0}^{N-1} \left(\bold{e}_k^T\bold{Q}\bold{e}_k+\bold{u}_k^T\bold{R}\bold{u}_k\right)+\bold{e}_N^T\bold{Q}\bold{e}_N,
\end{equation}
where the term $\bold{e}_k^T \bold{Q} \bold{e}_k$ penalizes the target state tracking error, and $\bold{u}_k^T \bold{R} \bold{u}_k$ accounts for the power cost incurred by generating the control input\cite{b10}. The weight matrices $\bold{Q}$ and $\bold{R}$ directly determine the dynamic performance of the system and the penalty associated with control energy. In particular, $\bold{Q}\succ0$ penalizes state deviations and reflects the relative importance of different state components, whereas $\bold{R}\succ0$ penalizes control effort.

Given the quadratic cost function in (\ref{LQR-1}), the objective of the LQR traget tracking controller is to determine the optimal sequence of control inputs $\{\bold{u}_0, \bold{u}_1, \cdots, \bold{u}_{N-1}\}$ that minimizes $\mathcal{J}(\bold{u}_k)$. To this end, dynamic programming is employed to derive the optimal control policy. Specifically, let $V_k(\bold{e})$ denote the optimal cost-to-go function from time slot $k$ to the terminal time slot $N$, which is defined as
\begin{equation}\label{LQR-2}\small
V_k(\bold{e}) = \min_{\{\bold{u}_k,\cdots,\bold{u}_{N-1}\}} \ \sum_{t=k}^{N-1} (\bold{e}_t^T\bold{Q}\bold{e}_t+\bold{u}_t^T\bold{R}\bold{u}_t) + \bold{e}_N^T\bold{Q}\bold{e}_N.\notag
\end{equation}
According to the Bellman optimality principle \cite{b30}, the following recursive relation holds:
\begin{equation}\label{LQR-3}\small
	V_k(\bold{e}) = \min_{\bold{u}_k} \Big\{
	\bold{e}_k^T\bold{Q}\bold{e}_k+\bold{u}_k^T\bold{R}\bold{u}_k + V_{k+1}(\bold{e})\Big\}.
\end{equation}
Under the linear state-feedback control law $\bold{u}_k = -\bold{K}_k \bold{e}_k$, there exists a symmetric positive semidefinite matrix $\bold{P}_k$ such that $V_k(\bold{e}) = \bold{e}_k^T\bold{P}_k\bold{e}_k$ \cite{b11}. By substituting $V_{k+1}(\bold{e}) = \bold{e}_{k+1}^T\bold{P}_{k+1}\bold{e}_{k+1}$ and $\bold{e}_{k+1} = \bold{A}\bold{e}_k + \bold{B}\bold{u}_k$ into (\ref{LQR-3}), we have
\vspace{-5pt}

\begin{subequations}\label{LQR-4}\small
	\begin{align}
		V_k(\bold{e}) 
		= \min_{\bold{u}_k} \Big\{
		&\bold{e}_k^T\bold{Q}\bold{e}_k+\bold{u}_k^T\bold{R}\bold{u}_k  \notag\\
		&+(\bold{A}\bold{e}_k+\bold{B}\bold{u}_k)^T
		\bold{P}_{k+1}(\bold{A}\bold{e}_k+\bold{B}\bold{u}_k)
		\Big\}.\tag{22}
	\end{align}
\end{subequations}
Differentiating the above expression with respect to $\bold{u}_k$ and setting the result to zero yields the optimal control law:
\begin{equation}\label{LQR-5}\small
\bold{u}_k^* = -\left(\bold{R}+\bold{B}^T\bold{P}_{k+1}\bold{B}\right)^{-1}\bold{B}^T\bold{P}_{k+1}\bold{A}\bold{e}_k.
\end{equation}
To compute the matrix $\bold{P}_{k+1}$, we substitute (\ref{LQR-5}) back into (\ref{LQR-4}), resulting in
\vspace{-5pt}

\begin{subequations}\label{D1}\small
\begin{align}
\bold{P}_k =& \bold{Q} + \bold{A}^T\bold{P}_{k+1}\bold{A} \notag\\
&-
\bold{A}^T\bold{P}_{k+1}\bold{B}(\bold{R}+\bold{B}^T\bold{P}_{k+1}\bold{B})^{-1}\bold{B}^T\bold{P}_{k+1}\bold{A},\tag{24}
\end{align}
\end{subequations}
which is known as the discrete-time algebraic Riccati equation (DARE). Based on this recursion, the sequence of matrices $\{\bold{P}_N,\cdots,\bold{P}_1\}$ can be computed offline using the terminal condition $V_N(\bold{e}) = \bold{e}_N^T \bold{Q}\bold{e}_N$. The corresponding optimal control law is then given by $\bold{u}_k^* = -\bold{K}_k\bold{e}_k$, where the feedback gain matrix is 
\begin{equation}\label{D2}\small
	\bold{K}_k = \left(\bold{R}+\bold{B}^T\bold{P}_{k+1}\bold{B}\right)^{-1}\bold{B}^T\bold{P}_{k+1}\bold{A}.
\end{equation}

However, the classical LQR framework is not directly applicable to the considered system for two main reasons. First, LQR is derived for unconstrained and deterministic linear systems, where the optimal feedback policy can be obtained in closed form via the Riccati recursion. The presence of nonlinear sensing and communication constraints destroys this algebraic structure, making it impossible to derive a valid closed-form feedback gain and thereby rendering the standard LQR solution inapplicable. Second, the horizon-wide cost in (\ref{LQR-1}) is unsuitable for our setting, because directly adopting it would lead to a global optimization problem that simultaneously involves $N$ sensing constraints and $N$ communication constraints, all of which are tightly coupled across time. This cross-horizon coupling significantly increases computational complexity and makes the resulting optimization problem impractical for real-time control.

To overcome these challenges, we propose an objective function that preserves the quadratic structure of LQR, incorporates the impact of stochastic noise through the expectation of future states, and adopts a receding-horizon optimization strategy inspired by model predictive control (MPC) over $N_0$ future time slots, where $N_0$ denotes the prediction horizon length, i.e.,
\begin{equation}\label{Obj} \small
	\mathcal{C}_{n}(\bold{u}_n)=\sum_{i=1}^{N_0}\ \left(\mathbb{E}_{\bold{n}_{s}}\left\{\bold{e}^T_{n+i}\bold{Q}\bold{e}_{n+i}\right\}+\bold{u}_{n+i-1}^T\bold{R}\bold{u}_{n+i-1}\right).\notag
\end{equation}
Accordingly, at the $n$-th time slot, the stochastic MPC problem for determining $\{\bold{u}_n,\bold{w}_{n+1}\}$ is formulated as
\vspace{-5pt}

\begin{subequations}\label{P1} \small
	\begin{alignat}{2}
		\underset{\{\bold{u}_{n+i-1},\bold{w}_{n+i}\}_{i\in\mathcal{I}}}{\text{min}} &\ \mathcal{C}_n(\bold{u}_n)\notag\\
		\text{s.t.}&\ R_{n+i} \ge R_{th},\forall i\in \mathcal{I},\label{P1-a}\\
		&\ \mathbb{E}_{\bold{n}_{s}}\left\{\frac{|\bold{a}^H_{n+i} \bold{w}_{n+i}|^2}{d_{n+i}^{4}}\right\} \ge \Gamma_{th},\forall i\in \mathcal{I},\label{P1-b}\\
		&\ \Vert\bold{w}_{n+i}\Vert^2 \le P_T,\forall i\in \mathcal{I},\label{P1-c}\\
		&\ \Vert\bold{u}_{n+i-1}\Vert \le a_{max}, \forall i\in \mathcal{I},\label{P1-d}\\
		&\ \Vert\bold{v}^{U}[n+i]\Vert\le V_{max}, \forall i\in \mathcal{I},\label{P1-e}
	\end{alignat}
\end{subequations}
where $\mathcal{I}\triangleq\{1,\cdots,N_0\}$, $d_{n+i}\triangleq d(\bold{p}^{U}[n+i],\bold{p}^{t}[n+i])$, and $\Gamma_{th}\triangleq\frac{\sigma_r^2}{GM_r\beta_r}\text{SNR}_{th}^{echo}$, with $\text{SNR}_{th}^{echo}$ being the threshold SNR of the received echo signal. Moreover, $R_{th}$ and $P_T$ denote the communication rate threshold and the transmit power budget, respectively. Constraint (\ref{P1-a}) guarantees the communication requirement, (\ref{P1-b}) ensures that the expected SNR of the received echo signal exceeds the threshold $\text{SNR}_{th}^{echo}$, (\ref{P1-c}) enforces the transmit power limit, and constraints (\ref{P1-d})-(\ref{P1-e}) impose physical limitations on the UAV motion by bounding its acceleration and velocity to $a_{max}$ and $V_{max}$, respectively.

In general, problem (\ref{P1}) is difficult to be solved efficiently due to the following challenges: 1) constraints (\ref{P1-a}) and (\ref{P1-b}) are highly non-convex and introduce strong coupling between the control variables $\{\bold{u}_{n+i-1}\}$ and the beamforming vectors $\{\bold{w}_{n+i}\}$; 2) the expectation term in (\ref{P1-b}) does not admit a closed-form expression; and 3) the true probability distributions of the system states $\{\bold{e}_{n+i}\}$ are unknown.

\section{Joint Optimization of Mobile Target Tracking Control and Wireless Beamforming}\label{sec:IV}
In this section, to address problem (\ref{P1}), we first apply an EKF method to estimate and predict the target state from the received sensing echoes. Using the current state estimate and the transmit steering vector approximated from the predicted target position, we then derive a closed-form sequence of optimal beamforming vectors for any given control input sequence by solving a sensing-centric feasibility problem. This closed-form characterization effectively decouples beamforming design from the control variables, thereby enabling an equivalent reformulation of the original problem into a more tractable control-driven one, which can subsequently be solved efficiently via relaxation-based convex approximation method.

Following this procedure at each time slot, the proposed ISCC framework is illustrated in Fig. \ref{fig2}. In particular, at the first time slot, the UAV jointly designs its control input and beamforming vector using the initial target state estimate acquired from external sensing devices.
\begin{figure}[t]
	\centering
	\includegraphics[width=0.95\linewidth]{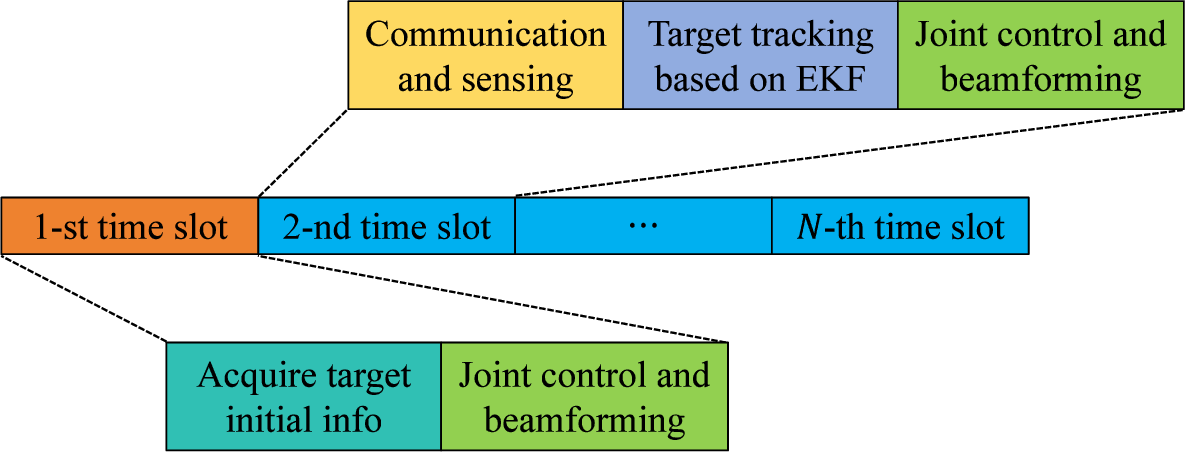}
	\caption{Illustration of the proposed ISCC framework.}
	\label{fig2}
	\vspace{0pt}
\end{figure}

\subsection{Extended Kalman Filtering}
Given the state evolution model in (\ref{S2}) and the measurement model in (13), we apply an EKF to estimate and predict the target state $\bold{s}^t[n]$. Since $\bold{f}(\bold{s}^t[n])$ is nonlinear with respect to $\bold{s}^t[n]$, we approximate it via linearization by computing its Jacobian matrix \cite{b6}, which is given in (\ref{Jacobian}) at the bottom of the next page, where $\Delta \bold{p}[n] \triangleq [\Delta p_x[n], \Delta p_y[n]]^T = \bold{p}^{U}[n] - \bold{p}^{t}[n]$, $\Delta \bold{v}[n] \triangleq [\Delta v_x[n], \Delta v_y[n]]^T = \bold{v}^{U}[n] - \bold{v}^{t}[n]$ and $\bm{\psi}_n \triangleq \bold{b}_n \bold{a}_n^H \bold{w}_n$. In (\ref{Jacobian}), it remains to compute $\frac{\partial (\bm{\psi}_nd_{n}^{-2})}{\partial p_x^t[n]}$ and $\frac{\partial (\bm{\psi}_nd_{n}^{-2})}{\partial p_y^t[n]}$. Using the chain rule, we obtain
\begin{equation}\label{noise}\small
	\left\{\begin{aligned}
		\frac{\partial (\bm{\psi}_nd_{n}^{-2})}{\partial p_x^t[n]} &= \frac{\partial \bm{\psi}_n}{\partial p_x^t[n]}d_{n}^{-2} + 2\Delta p_x[n]\bm{\psi}_nd_{n}^{-4},\notag\\
		\frac{\partial (\bm{\psi}_nd_{n}^{-2})}{\partial p_y^t[n]} &= \frac{\partial \bm{\psi}_n}{\partial p_y^t[n]}d_{n}^{-2} + 2\Delta p_y[n]\bm{\psi}_nd_{n}^{-4},\notag
	\end{aligned}
	\right.\notag
\end{equation}
where the expressions of $\frac{\partial \bm{\psi}_n}{\partial p_x^t[n]}$ and $\frac{\partial \bm{\psi}_n}{\partial p_y^t[n]}$ are given in (\ref{J_px}) and (\ref{J_py}), respectively, also at the bottom of the next page.
\begin{figure*}[b]  
	\centering
	\vspace*{0pt} 
	\hrulefill 
	
	\begin{minipage}{\textwidth}
		\centering
		\begin{equation}\label{Jacobian}\small
			\frac{\partial \bold{f}}{\partial \bold{s}^t[n]}=\begin{bmatrix}
				\frac{-\Delta p_x[n]}{d_n} & \frac{-\Delta p_y[n]}{d_n} & 0 & 0 \\
				\frac{2f_c\left(\Delta v_x[n]d_n^2+(\Delta\bold{p}[n])^T\Delta\bold{v}[n]\Delta p_x[n]\right)}{cd_n^3}, & \frac{2f_c\left(\Delta v_y[n]d_n^2+(\Delta\bold{p}[n])^T\Delta\bold{v}[n]\Delta p_y[n]\right)}{cd_n^3} & \frac{2f_c\Delta p_x[n]}{cd_n} & \frac{2f_c\Delta p_y[n]}{cd_n} \\
				\begin{bmatrix}
					\Re\left\{G\sqrt{\beta_r}\frac{\partial (\bm{\psi}_nd_{n}^{-2})}{\partial p_x^t[n]}\right\} \\
					\Im\left\{G\sqrt{\beta_r}\frac{\partial (\bm{\psi}_nd_{n}^{-2})}{\partial p_x^t[n]}\right\}
				\end{bmatrix} & \begin{bmatrix}
					\Re\left\{G\sqrt{\beta_r}\frac{\partial (\bm{\psi}_nd_{n}^{-2})}{\partial p_y^t[n]}\right\} \\
					\Im\left\{G\sqrt{\beta_r}\frac{\partial (\bm{\psi}_nd_{n}^{-2})}{\partial p_y^t[n]}\right\}
				\end{bmatrix} & \bold{0} & \bold{0} 
			\end{bmatrix}.
		\end{equation}
	\end{minipage}
	
	\vspace{2pt} 
	\hrulefill 
	
	\begin{minipage}{\textwidth}
		\centering
		\begin{equation}\label{J_px}\small
			\begin{aligned}
				\frac{\partial \bm{\psi}_n}{\partial p_x^t[n]} =& 
				j\pi \mathbf{b}_n \odot \Bigl( 
				\frac{-(\Delta p_y[n])^2-H^2}{d_n^3}([0,\dots,M_x^r-1]^T \otimes \mathbf{1}_{M_y^r})
				+\frac{\Delta p_x[n] \Delta p_y[n]}{d_n^3} (\mathbf{1}_{M_x^r} \otimes [0,\dots,M_y^r-1]^T)
				\Bigr) \mathbf{a}_n^H \mathbf{w}_n \\
				&+ \mathbf{b}_n
				\Biggl( j \pi\mathbf{a}_n\odot \Bigl(
				\frac{-(\Delta p_y[n])^2-H^2}{d_n^3} ([0,\dots,M_x^t-1]^T \otimes \mathbf{1}_{M_y^t})+\frac{\Delta p_x[n] \Delta p_y[n]}{d_n^3}(\mathbf{1}_{M_x^t} \otimes [0,\dots,M_y^t-1]^T )
				\Bigr) \Biggr)^H \mathbf{w}_n.
			\end{aligned}
		\end{equation}
	\end{minipage}
	
	\vspace{2pt}
	\hrulefill 
	
	\begin{minipage}{\textwidth}
		\centering
		\begin{equation}\label{J_py}\small
			\begin{aligned}
				\frac{\partial \bm{\psi}_n}{\partial p_y^t[n]} &= 
				j\pi \mathbf{b}_n \odot \Bigl( 
				\frac{\Delta p_x[n] \Delta p_y[n]}{d_n^3}([0,\dots,M_x^r-1]^T \otimes \mathbf{1}_{M_y^r})
				+\frac{-(\Delta p_x[n])^2-H^2}{d_n^3}(\mathbf{1}_{M_x^r} \otimes [0,\dots,M_y^r-1]^T)
				\Bigr) \mathbf{a}_n^H \mathbf{w}_n \\
				&\quad + \mathbf{b}_n
				\Biggl( j \pi\mathbf{a}_n\odot \Bigl(
				\frac{\Delta p_x[n] \Delta p_y[n]}{d_n^3}([0,\dots,M_x^t-1]^T \otimes \mathbf{1}_{M_y^t} )
				+\frac{-(\Delta p_x[n])^2-H^2}{d_n^3}(\mathbf{1}_{M_x^t} \otimes [0,\dots,M_y^t-1]^T )
				\Bigr) \Biggr)^H \mathbf{w}_n.
			\end{aligned}
		\end{equation}
	\end{minipage}
\end{figure*}

We are now ready to introduce the EKF scheme. For notational simplicity, let $\hat{\bold{s}}^t[n]$ and $\check{\bold{s}}^t[n]$ denote the estimated and predicted states of the target at the $n$-th time slot, respectively. Following the standard Kalman filtering procedure \cite{b7}, the target tracking design can be summarized as follows:

1) \textbf{State Prediction:}
\begin{equation}\label{EKF-1}\small 
	\check{\bold{s}}^t[n-1+i]=\bold{A}^i\hat{\bold{s}}^t[n-1], \forall i\in \mathcal{I},
\end{equation}

2) \textbf{Linearization:}
\begin{equation}\label{EKF-2}\small 
	\mathbf{F}_{n} = \left. \frac{\partial \bold{f}}{\partial \bold{s}^t[n]} \right|_{\bold{s}^t[n] = \check{\bold{s}}^t[n]},
\end{equation}

3) \textbf{MSE Matrix Prediction:}
\begin{equation}\label{EKF-3} \small
	\check{\mathbf{M}}_{n} = \mathbf{A} \hat{\mathbf{M}}_{n-1} \mathbf{A}^T + \mathbf{Q}_{s},
\end{equation}

4) \textbf{Kalman Gain Calculation:}
\begin{equation}\label{EKF-4} \small
	\mathbf{K}_{al,n} = \check{\mathbf{M}}_{n} \mathbf{F}_{n}^{H} \left(\mathbf{F}_{n} \check{\mathbf{M}}_{n} \mathbf{F}_{n}^{H} + \mathbf{Q}_{m} \right)^{-1},
\end{equation}

5) \textbf{State Tracking:}
\begin{equation}\label{EKF-5} \small
	\hat{\mathbf{s}}^t[n] = \check{\mathbf{s}}^t[n] + \mathbf{K}_{al,n} \left( \mathbf{m}_{n} - \mathbf{f} \left( \check{\mathbf{s}}^t[n] \right) \right),
\end{equation}

6) \textbf{MSE Matrix Update:}
\begin{equation}\label{EKF-6} \small
	\hat{\mathbf{M}}_{n} = \left( \mathbf{I} - \mathbf{K}_{al,n} \mathbf{F}_{n} \right) \check{\mathbf{M}}_{n},
\end{equation}
where $\check{\mathbf{M}}_{n}\triangleq \mathbb{E}\{(\bold{s}^t[n]-\check{\bold{s}}^t[n])(\bold{s}^t[n]-\check{\bold{s}}^t[n])^T\}$ and $\hat{\mathbf{M}}_{n}\triangleq \mathbb{E}\{(\bold{s}^t[n]-\hat{\bold{s}}^t[n])(\bold{s}^t[n]-\hat{\bold{s}}^t[n])^T\}$ denote the predicted and estimated error covariance matrices, respectively. In this EKF scheme, at the first time slot ($n=1$), no prediction of the target state or error covariance is required, as the target state estimate and its associated error covariance, $\{\hat{\mathbf{s}}^t[1], \hat{\mathbf{M}}_1\}$, are directly provided by external sensing devices. For subsequent slots ($n \ge 2$), the predicted target state and error covariance, $\{\check{\mathbf{s}}^t[n], \check{\mathbf{M}}_n\}$, together with the current measurement $\mathbf{m}_n$, are used to update both the state estimate $\hat{\mathbf{s}}^t[n]$ and its associated error covariance $\hat{\mathbf{M}}_n$. These updated quantities are then employed to characterize the distributions of the future target state tracking errors $\{\bold{e}_{n+i}\}_{i\in\mathcal{I}}$.

\subsection{Joint Control and Beamforming Design}
At the $n$-th time slot, after obtaining the state estimate $\hat{\mathbf{s}}^t[n]$ and its associated error covariance $\hat{\mathbf{M}}_n$ from the current measurement $\mathbf{m}_n$, our goal is to determine the control input $\bold{u}_n$ and the beamforming vector $\bold{w}_{n+1}$ by solving the MPC problem in (\ref{P1}). To this end, we model the target state $\bold{s}^t[n]$ as a Gaussian random variable with a known distribution, i.e.,
\begin{equation}\label{S-n+1}\small
	\bold{s}^t[n] = \check{\bold{s}}^t[n] + \bm{\xi}_{n}, \ \bm{\xi}_{n} \sim \mathcal{N}(\bold{0}, \hat{\mathbf{M}}_{n}).
\end{equation}
Based on (\ref{S-n+1}), the target state tracking error $\bold{e}_{n}$ can be expressed as
\begin{equation}\label{e}\small
	\bold{e}_{n}=\bold{s}^U[n]-\bold{s}^t[n]=\hat{\bold{e}}_n-\bm{\xi}_{n},
\end{equation}
where $\hat{\bold{e}}_n \triangleq \bold{s}^U[n]-\hat{\bold{s}}^t[n]$ denotes the estimated system state at the $n$-th time slot. By combining the state-space equation (\ref{S4}) with (\ref{e}), the distributions of the future target state tracking errors $\{\bold{e}_{n+i}\}_{i\in\mathcal{I}}$ can be characterized as
\vspace{-5pt}

\begin{subequations}\label{E-43}\small
\begin{align}
	\bold{e}_{n+i}=&\bold{A}^{i}\bold{e}_n+\sum_{j=0}^{i-1}\bold{A}^{i-j-1}(\bold{B}\bold{u}_{n+j}+\bold{n}_{s,j})\notag\\
	=&\bold{A}^{i}(\hat{\bold{e}}_n-\bm{\xi}_{n})+\sum_{j=0}^{i-1}\bold{A}^{i-j-1}(\bold{B}\bold{u}_{n+j}+\bold{n}_{s,j})\notag\\
	\triangleq& \bm{\Lambda}_{i}(\overline{\bold{u}}_{i}+\overline{\bold{n}}_{i}) ,\ \forall i\in\mathcal{I},\tag{38}
\end{align}
\end{subequations}
where 
\begin{equation}\small
	\bm{\Lambda}_{i}\triangleq \begin{bmatrix} \bold{A}^{i} & \bold{A}^{i-1} & \cdots & \bold{A} & \bold{I} \end{bmatrix},\notag
\end{equation}
\begin{equation}\small
	\overline{\bold{u}}_{i}\triangleq\begin{bmatrix} \hat{\bold{e}}_n^T& (\bold{B}\bold{u}_{n})^T & \cdots & (\bold{B}\bold{u}_{n+i-1})^T \end{bmatrix}^T,\notag
\end{equation}
\begin{equation}
	\overline{\bold{n}}_{i}\triangleq\begin{bmatrix} -\bm{\xi}_{n}^T & \bold{n}_{s,0}^T & \cdots & \bold{n}_{s,i-1}^T \end{bmatrix}^T.\notag
\end{equation}
Here, $\bold{n}_{s,j}$ denotes the process noise at the $(n+j)$-th time slot. Each $\bold{n}_{s,j}$ follows the same distribution as $\bold{n}_s$ in (\ref{S4}) and is independent across different time slots.

Although the distributions of target state tracking errors $\{\bold{e}_{n+i}\}_{i\in\mathcal{I}}$ are available, handling constraint (\ref{P1-b}) remains challenging because the noise term $\overline{\bold{n}}_{i}$ appears in both the transmit steering vector $\bold{a}_{n+i}$ and the denominator $d^4_{n+i}$. To mitigate this difficulty, we propose to approximate $\bold{a}_{n+i}$ by $\check{\bold{a}}_{n+i}\triangleq \bold{a}(\bold{p}^{U}[n+i],\check{\bold{p}}^{t}[n+i])$, where $\check{\bold{p}}^{t}[n+i]$ denotes the predicted target position at the $(n+i)$-th time slot. This approximation enables a closed-form representation of the sensing constraint. With this modification, problem (\ref{P1}) can be reformulated as
\vspace{-5pt}

\begin{subequations}\label{P2}\small
	\begin{alignat}{2}
		\underset{\{\bold{u}_{n+i-1},\bold{w}_{n+i}\}_{i\in\mathcal{I}}}{\text{min}} &\ \mathcal{C}_n(\bold{u}_n)\notag\\
		\text{s.t.}&\ \text{(\ref{P1-a})},\text{(\ref{P1-c})}\text{-}\text{(\ref{P1-e})},\notag\\
		&\ \frac{\vert\check{\bold{a}}_{n+i}\bold{w}_{n+i}\vert^2}{\mathbb{E}_{\overline{\bold{n}}_{i}}\left\{d_{n+i}^{4}\right\}} \ge \Gamma_{th}, \forall i\in\mathcal{I}.\label{P2-b}
	\end{alignat}
\end{subequations}
However, problem (\ref{P2}) is still challenging to solve due to the non-convex constraints (\ref{P1-a}) and (\ref{P2-b}).

To proceed, we observe that the objective function of problem (\ref{P2}) depends solely on the control input variables $\{\bold{u}_{n-1+i}\}_{i\in\mathcal{I}}$. Therefore, once these control inputs are specified, it suffices to determine whether there exists a corresponding set of beamforming vectors $\{\bold{w}_{n+i}\}_{i\in\mathcal{I}}$ that satisfies constraints (\ref{P1-a}) and (\ref{P1-c}) in problem (\ref{P2}). Conversely, the feasibility of $\{\bold{u}_{n-1+i}\}_{i\in\mathcal{I}}$ can be assessed by verifying the existence of such a beamforming vector set. Motivated by this observation, we can formulate feasibility-check problems for any specified control inputs. Depending on the performance metric of interest under the given control inputs $\{\bold{u}_{n-1+i}\}_{i\in\mathcal{I}}$, the feasibility-check formulation can be categorized into two cases, depending on whether the performance metric is sensing-centric or communication-centric.

\textbf{Sensing-centric feasibility check:} 
In this case, the feasibility-check problem is formulated as 
\vspace{-5pt}

\begin{subequations}\label{S-check}\small
	\begin{alignat}{2}
		\underset{\{\bold{w}_{n+i}\}_{i\in\mathcal{I}}}{\text{max}} &\ \Gamma\notag\\
		\text{s.t.}&\ \text{(\ref{P1-a})},\text{(\ref{P1-c})},\notag\\
		&\ \frac{\vert\check{\bold{a}}_{n+i}\bold{w}_{n+i}\vert^2}{\mathbb{E}_{\overline{\bold{n}}_{i}}\left\{d_{n+i}^{4}\right\}}- \Gamma_{th} \ge \Gamma, \forall i\in\mathcal{I}.\label{Sc-b}
	\end{alignat}
\end{subequations}
Note that in problem (\ref{S-check}), the optimization variables $\{\bold{w}_{n+1}, \dots, \bold{w}_{n+N_0}\}$ are mutually independent across different time slots. Consequently, it can be decomposed into $N_0$ parallel subproblems, each associated with a single time slot. The optimal objective value of the original problem is therefore given by the minimum among the optimal objective values of these $N_0$ subproblems. Specifically, the $i$-th subproblem can be formulated as
\vspace{-5pt}

\begin{subequations}\label{S-sub}\small
	\begin{alignat}{2}
		\underset{\bold{w}_{n+i}}{\text{max}} \ & |\check{\bold{a}}_{n+i}^H \bold{w}_{n+i}|^2 - \mathbb{E}_{\overline{\bold{n}}_{i}} \big\{d_{n+i}^4\big\} \Gamma_{th} \notag\\
		\text{s.t.}\ & R_{n+i} \ge R_{th}, \label{S-sub-a}\\
		& \Vert \bold{w}_{n+i}\Vert^2 \le P_T. \label{S-sub-b}
	\end{alignat}
\end{subequations}

Although problem (\ref{S-sub}) is non-convex, its optimal solution can be derived in closed form as\cite{I12}
\begin{equation}\label{E-49}\small
	\bold{w}^*_{n+i}=
	\begin{cases}
		\sqrt{P_T}\frac{\check{\bold{a}}_{n+i}}{\Vert\check{\bold{a}}_{n+i}\Vert}, &\text{if } P_T|\bold{a}_{c,i}^H\check{\bold{a}}_{n+i}|^2 \geq M_t\eta d_{c,i}^2, \\
		\kappa_1\bm{\nu}_1+\kappa_2\bm{\nu}_2,&\text{otherwise},
	\end{cases}
\end{equation}
where $\bold{a}_{c,i}\triangleq\bold{a}(\bold{p}^{U}[n+i],\bold{p}^{u})$, $d_{c,i}\triangleq d(\bold{p}^{U}[n+i],\bold{p}^u)$ and $\eta \triangleq \frac{\sigma_c^2(2^{R_{th}}-1)}{\beta_0}$. Moreover, the auxiliary terms $\kappa_1,\kappa_2,\bm{\nu}_1,\bm{\nu}_2$ are defined as 
\begin{equation}\small
	\bm{\nu}_1=\frac{\bold{a}_{c,i}}{\Vert\bold{a}_{c,i}\Vert},\bm{\nu}_2=\frac{\check{\bold{a}}_{n+i}-(\bm{\nu}^H_1\check{\bold{a}}_{n+i})\bm{\nu}_1}{\Vert\check{\bold{a}}_{n+i}-(\bm{\nu}^H_1\check{\bold{a}}_{n+i})\bm{\nu}_1\Vert},\notag
\end{equation}
\begin{equation}\small
	\kappa_1=\sqrt{\frac{\eta d_{c,i}^2}{M_t}}\frac{\bm{\nu}^H_1\check{\bold{a}}_{n+i}}{|\bm{\nu}^H_1\check{\bold{a}}_{n+i}|},\kappa_2=\sqrt{P_T-\frac{\eta d_{c,i}^2}{M_t}}\frac{\bm{\nu}^H_2\check{\bold{a}}_{n+i}}{|\bm{\nu}^H_2\check{\bold{a}}_{n+i}|}.\notag
\end{equation}

By substituting $\bold{w}^*_{n+i}$ in (\ref{E-49}) back into problem (\ref{S-sub}), the optimal achievable value of the beamforming-dependent term $\Gamma^*_{n+i}\triangleq\vert\check{\bold{a}}_{n+i}\bold{w}_{n+i}^*\vert^2$ is given by
\begin{equation}\label{Optimal}\small
	\Gamma^*_{n+i}=
	\begin{cases}
		\gamma, \ \ \text{if} \ \gamma\cos^2\theta_{i} \geq \eta d_{c,i}^2, \\
		\left(\sqrt{\eta d_{c,i}^2}\cos\theta_{i}+\sqrt{\gamma-\eta d_{c,i}^2}\sin\theta_{i}\right)^2,\text{otherwise},
	\end{cases}
\end{equation}
where $\gamma \triangleq M_tP_T$ and $\theta_{i}\triangleq\arccos\frac{\vert\check{\bold{a}}_{n+i}^H\bold{a}_{c,i}\vert}{\Vert\check{\bold{a}}_{n+i}\Vert\Vert\bold{a}_{c,i}\Vert}\in \left[0,\frac{\pi}{2}\right]$. Therefore, the optimal objective value of the overall sensing-centric feasibility-check problem (\ref{S-check}) can be expressed as
\begin{equation}\label{S-Gamma}\small
	\Gamma^* = \min_{i\in\mathcal{I}} \Big\{ \Gamma^*_{n+i} - \mathbb{E}_{\overline{\bold{n}}_{i}}\{d_{n+i}^4\} \Gamma_{th} \Big\},
\end{equation}  
and the control inputs $\{\bold{u}_{n-1+i}\}_{i\in\mathcal{I}}$ are feasible if and only if $\Gamma^* \ge 0$. 

\textbf{Communication-centric feasibility check:} 
In this case, the feasibility-check problem is formulated as 
\vspace{-5pt}

\begin{subequations}\label{C-check}\small
	\begin{alignat}{2}
		\underset{\{\bold{w}_{n+i}\}_{i\in\mathcal{I}}}{\text{max}} &\ R\notag\\
		\text{s.t.}&\ \text{(\ref{P1-c})},\text{(\ref{P2-b})},\notag\\
		&\ R_{n+i} - R_{th} \ge R,\forall i\in \mathcal{I}.\label{Cc-b}
	\end{alignat}
\end{subequations}
Similar to the sensing-centric case, problem (\ref{C-check}) can be decomposed into $N_0$ parallel subproblems, each admitting a closed-form optimal solution. The optimal value of problem (\ref{C-check}), denoted by $R^*$, is the minimum among the optimal values of these subproblems. Therefore, the control inputs $\{\bold{u}_{n-1+i}\}_{i\in\mathcal{I}}$ are feasible in the communication-centric case if and only if $R^* \ge 0$. For completeness, the detailed derivations of the optimal solution for problem (\ref{C-check}) are provided in Appendix A.

For a given set of control inputs $\{\bold{u}_{n-1+i}\}_{i\in\mathcal{I}}$, the sensing-centric and communication-centric feasibility-check problems adopt different evaluation criteria. To clarify the relationship between these two approaches and guide the selection of an appropriate check, we present the following Lemma \ref{lemma-1}.
\begin{lemma}\label{lemma-1}
	For any given control inputs $\{\bold{u}_{n-1+i}\}_{i\in\mathcal{I}}$, the sensing-centric and communication-centric feasibility-check problems, i.e., (\ref{S-check}) and (\ref{C-check}), yield identical outcomes, i.e.,
\begin{equation}\label{E-55}\small
	R^* \ge 0  \Longleftrightarrow  \Gamma^* \ge 0.
\end{equation}
\end{lemma}
\begin{IEEEproof}
	Please see Appendix B.
\end{IEEEproof}

Based on Lemma \ref{lemma-1}, we adopt the check criterion $\Gamma^* \ge 0$ to equivalently reformulate problem (\ref{P2}) as
\vspace{-5pt}

\begin{subequations}\label{P3}\small
	\begin{alignat}{2}
		\underset{\{\bold{u}_{n+i-1}\}_{i\in\mathcal{I}}}{\text{min}} &\ \mathcal{C}_n(\bold{u}_n)\notag\\
		\text{s.t.}	&\ \text{(\ref{P1-d})},\text{(\ref{P1-e})},\notag\\
		&\ \Gamma^*_{n+i} - \mathbb{E}_{\overline{\bold{n}}_{i}}\{d_{n+i}^4\} \Gamma_{th} \ge 0, \forall i\in\mathcal{I}.\label{P3-b}
	\end{alignat}
\end{subequations}
However, problem (\ref{P3}) remains challenging due to the piece-wise non-concave nature of $\Gamma^*_{n+i}$. To address this issue, we derive a lower bound of $\Gamma^*_{n+i}$ as follows.
\begin{lemma}\label{lemma-2}
	The optimal achievable value of the beamforming-dependent term $\Gamma^*_{n+i}$ in (\ref{Optimal}) admits the following lower bound:
	\begin{equation}\label{lower}\small
		\Gamma^*_{n+i} \ge \gamma - \delta_i\eta d^2_{c,i} \triangleq \Gamma^l_{n+i},\ \forall i\in\mathcal{I},
	\end{equation}
where $\delta_i$ is a binary constant defined as
\begin{equation}\label{0-1}\small
	\delta_i = 
	\begin{cases}
		0, \ \text{if}\ \ \check{\bold{p}}^{t}[n+i]=\bold{p}^u,\\
		1,\ \text{otherwise}.
	\end{cases}\notag
\end{equation}
Moreover, the lower bound $\Gamma^l_{n+i}$ in (\ref{lower}) becomes asymptotically tight as the number of transmit antennas $M_t\rightarrow\infty$.
\end{lemma}
\begin{IEEEproof}
	Please see Appendix C.
\end{IEEEproof}

Based on Lemma \ref{lemma-2}, a relaxed version of problem (\ref{P3}) is obtained by replacing $\Gamma^*_{n+i}$ with $\Gamma^l_{n+i}$, i.e.,
\vspace{-5pt}

\begin{subequations}\label{P4}\small
	\begin{alignat}{2}
		\underset{\{\bold{u}_{n+i-1}\}_{i\in\mathcal{I}}}{\text{min}} &\ \sum_{i=1}^{N_0}\ \left(\mathbb{E}_{\overline{\bold{n}}_{i}}\left\{\bold{e}^T_{n+i}\bold{Q}\bold{e}_{n+i}\right\}+\bold{u}_{n+i-1}^T\bold{R}\bold{u}_{n+i-1}\right)\notag\\
		\text{s.t.}	&\ \text{(\ref{P1-d})},\text{(\ref{P1-e})},\notag\\
		&\ \delta_i\eta d^2_{c,i}+\mathbb{E}_{\overline{\bold{n}}_{i}}\{d_{n+i}^4\} \Gamma_{th} -\gamma \le 0, \forall i\in\mathcal{I}.\label{P4-b}
	\end{alignat}
\end{subequations}
Accordingly, a high-quality solution of the original problem (\ref{P3}) can be obtained by solving the relaxed problem (\ref{P4}). 

To this end, we ought to derive the closed-form expressions of $\mathbb{E}_{\overline{\bold{n}}_{i}}\left\{\bold{e}^T_{n+i}\bold{Q}\bold{e}_{n+i}\right\}$, $\mathbb{E}_{\overline{\bold{n}}_{i}}\{d_{n+i}^4\}$ and $d^2_{c,i}$. First, based on the distribution of $\bold{e}_{n+i}$ given in (\ref{E-43}), we have
\vspace{-5pt}

\begin{subequations}\label{E-59}\small
	\begin{alignat}{2}
		\mathbb{E}_{\overline{\bold{n}}_{i}}\left\{\bold{e}^T_{n+i}\bold{Q}\bold{e}_{n+i}\right\}&=\mathbb{E}_{\overline{\bold{n}}_{i}}\left\{\left(\bm{\Lambda}_{i}(\overline{\bold{u}}_{i}+\overline{\bold{n}}_{i})\right)^T\bold{Q}\bm{\Lambda}_{i}(\overline{\bold{u}}_{i}+\overline{\bold{n}}_{i})\right\}\notag\\
		&=\overline{\bold{u}}_{i}^T\bm{\Lambda}_{i}^T\bold{Q}\bm{\Lambda}_{i}\overline{\bold{u}}_{i} + \mathbb{E}_{\overline{\bold{n}}_{i}}\left\{\overline{\bold{n}}_{i}^T\bm{\Lambda}_{i}^T\bold{Q}\bm{\Lambda}_{i}\overline{\bold{n}}_{i}\right\}\notag\\
		&=\overline{\bold{u}}_{i}^T\bm{\Lambda}_{i}^q\overline{\bold{u}}_{i} + \text{Tr}\left(\overline{\bold{N}}_{i}\bm{\Lambda}_{i}^q\right),\tag{50}
	\end{alignat}
\end{subequations}
where $\bm{\Lambda}_{i}^q\triangleq\bm{\Lambda}_{i}^T\bold{Q}\bm{\Lambda}_{i}\succeq 0$ and $\overline{\bold{N}}_{i}\triangleq\mathbb{E}_{\overline{\bold{n}}_{i}}\left\{\overline{\bold{n}}_{i}\overline{\bold{n}}_{i}^T\right\}\succeq 0$. Next, since $d_{n+i}^2=\Vert\bold{p}^{U}[n+i]-\bold{p}^{t}[n+i]\Vert^2+H^2$, the fourth-order moment $\mathbb{E}_{\overline{\bold{n}}_{i}}\{d_{n+i}^4\}$ can be expressed in closed form as shown in (\ref{E-60}) at the bottom of this page,
\begin{figure*}[b]  
	\centering
	\vspace*{0pt} 
	\hrulefill 
	
	\begin{minipage}{\textwidth}
		\centering
		\begin{subequations}\label{E-60}\small
			\begin{alignat}{2}
				\mathbb{E}_{\overline{\bold{n}}_{i}}\left\{d_{n+i}^4\right\}
				=& \mathbb{E}_{\overline{\bold{n}}_{i}}\Big\{\big(\Vert\bold{p}^{U}[n+i]-\bold{p}^{t}[n+i]\Vert^2+H^2\big)^2\Big\}\notag\\
				=& (\overline{\bold{u}}_{i}^T \bm{\Lambda}_{i}^c \overline{\bold{u}}_{i})^2 
				+ 4 \overline{\bold{u}}_{i}^T \bm{\Lambda}_{i}^c \overline{\bold{N}}_{i} \bm{\Lambda}_{i}^c \overline{\bold{u}}_{i} 
				+ 2 (\overline{\bold{u}}_{i}^T \bm{\Lambda}_{i}^c \overline{\bold{u}}_{i}) \mathrm{Tr}(\bm{\Lambda}_{i}^c \overline{\bold{N}}_{i}) \notag\\
				&+ (\mathrm{Tr}(\bm{\Lambda}_{i}^c \overline{\bold{N}}_{i}))^2 
				+ 2 \mathrm{Tr}\big((\bm{\Lambda}_{i}^c \overline{\bold{N}}_{i})^2\big)
				+ 2 H^2 \Big(\overline{\bold{u}}_{i}^T \bm{\Lambda}_{i}^c \overline{\bold{u}}_{i} + \mathrm{Tr}(\bm{\Lambda}_{i}^c \overline{\bold{N}}_{i})\Big)
				+ H^4.	\tag{51}
			\end{alignat}
		\end{subequations}
	\end{minipage}
	\vspace{-5pt}
\end{figure*}
where $\bm{\Lambda}_{i}^c\triangleq\bm{\Lambda}_{i}^T\bold{C}^T\bold{C}\bm{\Lambda}_{i}\succeq 0$, and the constant matrix $\bold{C}$ is defined as
\begin{equation}\label{E-61}\small
	\bold{C}=\begin{bmatrix}
		1 & 0 & 0 & 0 \\
		0 & 1 & 0 & 0
	\end{bmatrix}.\notag
\end{equation}
Finally, to derive a closed-form expression for $d_{c,i}^2$, we express the UAV position $\bold{p}^U[n+i]$ using (\ref{S1}) as
\vspace{-5pt}

\begin{subequations}\label{E-UAV}\small
\begin{align}
	\bold{p}^U[n+i] &= \bold{C}\left(\bold{A}^{i} \bold{s}^U[n] + \sum_{j=0}^{i-1} \bold{A}^{i-j-1} \bold{B} \bold{u}_{n+j}\right) \notag\\
	&= \bold{C}\bm{\Lambda}_{i} \tilde{\bold{u}}_{i}, \ \ \forall i \in\mathcal{I},\tag{52}
\end{align}
\end{subequations}
where
$\tilde{\bold{u}}_{i} \triangleq
	\begin{bmatrix} 
		(\bold{s}^U[n])^T& (\bold{B} \bold{u}_{n})^T &\cdots & (\bold{B} \bold{u}_{n+i-1})^T\end{bmatrix}^T$. Based on (\ref{E-UAV}), $d_{c,i}^2$ can be expressed as
\begin{subequations}\label{E-63}\small
	\begin{align}
		d_{c,i}^2 
		&= \Vert \bold{C} \bm{\Lambda}_{i} \tilde{\bold{u}}_{i} - \bold{p}^{u} \Vert^2 + H^2 \notag\\
		&= \tilde{\bold{u}}_{i}^T \bm{\Lambda}_{i}^c \tilde{\bold{u}}_{i} 
		- 2 \tilde{\bold{u}}_{i}^T \bm{\Lambda}_{i}^T \bold{C}^T \bold{p}^{u} 
		+ (\bold{p}^{u})^T \bold{p}^{u} + H^2. \tag{53}
	\end{align}
\end{subequations}

To make the structure of problem (\ref{P4}) more explicit, we stack all control inputs $\{\bold{u}_{n-1+i}\}_{i\in\mathcal{I}}$ into a single vector, i.e.,
\begin{equation}\label{E-64}\small
	\hat{\bold{u}}_{n} \triangleq 
	\begin{bmatrix} \bold{u}_{n}^T & 
		\bold{u}_{n+1}^T&
		\cdots & 
		\bold{u}_{n+N_0-1}^T 
	\end{bmatrix}^T.
\end{equation}
With (\ref{E-64}), both $\overline{\bold{u}}_{i}$ and $\tilde{\bold{u}}_{i}$ can be expressed as the following known linear selections of $\hat{\bold{u}}_{n}$:
\begin{equation}\label{E-65}\small
	\overline{\bold{u}}_{i} = \bold{S}_{i}\hat{\bold{u}}_{n}+\overline{\bold{c}}_{i},\ \tilde{\bold{u}}_{i} = \bold{S}_{i}\hat{\bold{u}}_{n}+\tilde{\bold{c}}_{i},\forall i\in \mathcal{I},
\end{equation}
where $\overline{\mathbf {c}}_{i}=
\begin{bmatrix}\hat{\mathbf e}_n\\[4pt]\mathbf 0_{\,4 i}\end{bmatrix} \in \mathbb{R}^{(4+4i)\times 1}$, $\tilde{\mathbf c}_{i}=
\begin{bmatrix}\mathbf s^{U}[n]\\[4pt]\mathbf 0_{\,4 i}\end{bmatrix}\in \mathbb{R}^{(4+4i)\times 1}$, and the selection matrix $\bold{S}_{i}$ is given by
\begin{equation}\label{E-66}\small
	\bold{S}_{i}=
	\begin{bmatrix}
		\mathbf 0_{\,4\times (2 N_0)}\\[6pt]
		\big(\bold{I}_{i}\otimes \mathbf B\big)\quad  \mathbf 0_{\,4 i \times 2 (N_0-i)}
	\end{bmatrix}\in \mathbb{R}^{(4+4i) \times 2N_0}.\notag
\end{equation}
Then, by substituting (\ref{E-59}), (\ref{E-UAV}), (\ref{E-63}) and (\ref{E-65}) into (\ref{P4}) and performing some proper algebraic simplifications, problem (\ref{P4}) is reformulated as
\vspace{-5pt}

\begin{subequations}\label{P5}\small
	\begin{alignat}{2}
		\underset{\hat{\bold{u}}_{n}}{\text{min}} &\ \hat{\bold{u}}_{n}^{T} \bm{\Upsilon} \hat{\bold{u}}_{n}
		+ 2\mathbf{g}^{T} \hat{\bold{u}}_{n} + c_t,\notag\\
		\text{s.t.}&\ \Vert \mathbf{E}_i \hat{\bold{u}}_n \Vert \le a_{max}, \ \forall i \in \mathcal{I}, \label{P5-b}\\
		&\ \Vert \bold{v}^{U}[n] + \delta_t \mathbf{E}_{i,sum} \hat{\bold{u}}_n \Vert \le V_{max}, \ \forall i \in \mathcal{I}, \label{P5-c}\\
		&\ \Gamma_{th}\left((\bold{S}_{i}\hat{\bold{u}}_{n}+\overline{\bold{c}}_{i})^T \bm{\Lambda}_{i}^c (\bold{S}_{i}\hat{\bold{u}}_{n}+\overline{\bold{c}}_{i})\right)^2+ \hat{\bold{u}}_{n}^{T} \bm{\Xi}_{i} \hat{\bold{u}}_{n} \notag\\
		&\quad+ 2\bm{\zeta}_{i}^{T} \hat{\bold{u}}_{n} + \varpi_{i}\le 0,\  \forall i \in\mathcal{I},\label{P5-a}
	\end{alignat}
\end{subequations}
where $\mathbf{E}_i\triangleq \begin{bmatrix} \bold{0}_{2\times (2i-2)} &
	\bold{I}_{2}&
	\bold{0}_{2\times (2N_s-2i)} 
\end{bmatrix}\in \mathbb{R}^{2\times 2N_s}$, $\mathbf{E}_{i,sum}\triangleq\left(\sum_{j=1}^{i} \mathbf{E}_{j}\right)$, $\bm{\Upsilon} \triangleq \sum_{i=1}^{N_0} \mathbf{S}_{i}^{T} \bm{\Lambda}_{i}^{q} \mathbf{S}_{i}+ \mathbf{I}_{N_0} \otimes \mathbf{R}$, $\mathbf{g} \triangleq \sum_{i=1}^{N_0} \mathbf{S}_{i}^{T} \bm{\Lambda}_{i}^{q} \overline{\mathbf{c}}_{i}$ and $c_t \triangleq \sum_{i=1}^{N_0} \Big( \overline{\mathbf{c}}_{i}^{T} \bm{\Lambda}_{i}^{q}\overline{\mathbf{c}}_{i} 
+ \mathrm{Tr}(\overline{\mathbf{N}}_{i} \bm{\Lambda}_{i}^{q}) \Big)$. Moreover, the coefficients $\bm{\Xi}_{i}$, $\bm{\zeta}_{i}$ and $\varpi_{i}$ in constraint (\ref{P5-a}) are defined in (\ref{E-70}), (\ref{E-71}), and (\ref{E-72}), respectively, as shown at the bottom of the next page.

\begin{figure*}[b]  
	\centering
	\vspace*{0pt} 
	\hrulefill 
	
	\begin{minipage}{\textwidth}
		\centering
		\begin{subequations}\label{E-70}\small
			\begin{alignat}{2}
				\bm{\Xi}_i &= \delta_i\eta \mathbf{S}_i^T \bm{\Lambda}_i^c \mathbf{S}_i 
				+ 4 \Gamma_{th} \mathbf{S}_i^T \bm{\Lambda}_i^c \overline{\mathbf{N}}_i \bm{\Lambda}_i^c \mathbf{S}_i
				+ 2 \Gamma_{th} \left(H^2+\mathrm{Tr}(\bm{\Lambda}_i^c \overline{\mathbf{N}}_i)\right) \mathbf{S}_i^T \bm{\Lambda}_i^c \mathbf{S}_i.\tag{57}
			\end{alignat}
		\end{subequations}
	\end{minipage}
	
	\vspace{2pt} 
	\hrulefill 
	
	\begin{minipage}{\textwidth}
		\centering
		\begin{subequations}\label{E-71}\small
			\begin{alignat}{2}
				\bm{\zeta}_i &= \delta_i\eta \mathbf{S}_i^T (\bm{\Lambda}_i^c \tilde{\mathbf{c}}_i - \bm{\Lambda}_i^T \mathbf{C}^T \mathbf{p}^u)
				+ 4 \Gamma_{th} \mathbf{S}_i^T \bm{\Lambda}_i^c \overline{\mathbf{N}}_i \bm{\Lambda}_i^c \overline{\mathbf{c}}_i
				+ 2 \Gamma_{th} \left(H^2+\mathrm{Tr}(\bm{\Lambda}_i^c \overline{\mathbf{N}}_i)\right) \mathbf{S}_i^T \bm{\Lambda}_i^c \overline{\mathbf{c}}_i.\tag{58}
			\end{alignat}
		\end{subequations}
	\end{minipage}
	
	\vspace{2pt} 
	\hrulefill 
	
	\begin{minipage}{\textwidth}\small
		\centering
		\begin{subequations}\label{E-72}
			\begin{alignat}{2}
			\varpi_i &=\Gamma_{th} \Big[4 \overline{\mathbf{c}}_i^T \bm{\Lambda}_i^c \overline{\mathbf{N}}_i \bm{\Lambda}_i^c \overline{\mathbf{c}}_i 
			+ 2 (\overline{\mathbf{c}}_i^T \bm{\Lambda}_i^c \overline{\mathbf{c}}_i) \mathrm{Tr}(\bm{\Lambda}_i^c \overline{\mathbf{N}}_i)  + (\mathrm{Tr}(\bm{\Lambda}_i^c \overline{\mathbf{N}}_i))^2 
			+ 2 \mathrm{Tr}((\bm{\Lambda}_i^c \overline{\mathbf{N}}_i)^2) 
			+ 2 H^2 (\overline{\mathbf{c}}_i^T \bm{\Lambda}_i^c \overline{\mathbf{c}}_i + \mathrm{Tr}(\bm{\Lambda}_i^c \overline{\mathbf{N}}_i))+H^4\Big]\notag\\
			&\quad+\delta_i\eta (\tilde{\mathbf{c}}_i^T \bm{\Lambda}_i^c \tilde{\mathbf{c}}_i - 2 \tilde{\mathbf{c}}_i^T \bm{\Lambda}_i^T \mathbf{C}^T \mathbf{p}^u + (\mathbf{p}^u)^T \mathbf{p}^u + H^2)- \gamma.\tag{59}
			\end{alignat}
		\end{subequations}
	\end{minipage}
\end{figure*}
For problem (\ref{P5}), we have the following Theorem \ref{theorem-1}.
\begin{theorem}\label{theorem-1}
	Problem (\ref{P5}) is convex.
\end{theorem}
\begin{IEEEproof}
	Please see Appendix D.
\end{IEEEproof}

Due to the convexity of (\ref{P5}), one can efficiently solve it by using existing numerical tools, e.g., CVX\cite{b40}. After obtaining the optimal control input $\bold{u}_n^*$, a feasible transmit beamforming vector $\mathbf{w}_{n+1}^*$ can then be computed according to (\ref{E-49}). The proposed algorithm for solving problem (\ref{P1}) at each time slot is summarized in Algorithm \ref{ISCC}.
\begin{algorithm}[t]
	\caption{Proposed Algorithm for Solving Problem (\ref{P1})}
	\label{ISCC}
	\begin{algorithmic}[1]
		\STATE \textbf{Input:} Horizon length $N_0$; time slot index $n$; sensing measurement $\bold m_n$ and prior target prediction $\{\check{\bold s}^t[n],\check{\bold M}_n\}$ for $n\ge 2$; initial target estimate $\{\hat{\bold s}^t_{init},\hat{\bold M}_{init}\}$ for $n=1$.
		\STATE \textbf{Output:} Updated prediction $\{\check{\bold s}^t[n+1],\check{\bold M}_{n+1}\}$, control input $\bold u_n^{*}$, and feasible beamforming vector $\bold w_{n+1}^{*}$.
			\IF{$n=1$}
			\STATE Set $\{\hat{\bold s}^t[n],\hat{\bold M}_n\}\leftarrow\{\hat{\bold s}^t_{init},\hat{\bold M}_{init}\}$.
			\ELSE
			\STATE Obtain the posterior estimate $\{\hat{\bold s}^{t}[n],\hat{\bold M}_n\}$ via (\ref{EKF-5}) and (\ref{EKF-6}) with the prior $\{\check{\bold s}^{t}[n],\check{\bold M}_n\}$ and measurement $\bold m_n$.
			\ENDIF
		
			\STATE Predict the target states $\{\check{\bold s}^{t}[n+i]\}_{i\in\mathcal I}$ via (\ref{EKF-1}).
			\STATE Update the one-step prediction covariance $\check{\bold M}_{n+1}$ via (\ref{EKF-3}).
			\STATE Solve problem (\ref{P5}) to obtain the optimal control input $\bold u_n^*$.
			\STATE Compute a feasible transmit beamforming vector $\bold w_{n+1}^\star$ according to (\ref{E-49}).
		\end{algorithmic}
	\end{algorithm}

\section{Numerical Results}\label{sec:V}
In this section, we present numerical results to validate the
effectiveness of the proposed ISCC framework. Unless otherwise specified, the initial target and GU positions are set as $\bold{p}^{t}[1]=(0,400)$ m and $\bold{p}^{u}=(300,50)$ m. The UAV flies at a fixed altitude of $H=50$ m and is equipped with a UPA of $M_x^{t/r}=M_y^{t/r}=4$ antennas per axis, with a transmit power budget of $P_T=30$ dBm. The carrier frequency and matched-filtering gain are $f_c=30$ GHz and $G=10^3$, and the communication and radar noise powers are $\sigma_c^2=\sigma_r^2=-80$ dBm. The mission horizon is discretized into $N=300$ time slots with duration $\delta_t=0.2$ s, and the prediction horizon length is $N_0=5$. The weighting matrices are $\bold{Q}=\bold{I}_4$ and $\bold{R}=\bold{I}_2$, and the sensing constants are $a_1=20$ and $a_2=100$. The target motion noise variances are $\sigma_x^2=\sigma_y^2=4\times10^{-4}$ and $\sigma_{v_x}^2=\sigma_{v_y}^2=0.01$, while the echo SNR and communication rate thresholds are set to $\text{SNR}_{th}^{echo}=5$ dB and $R_{th}=2.5$ bps/Hz. We compare the proposed ISCC scheme with the following benchmark methods:
\begin{itemize}
	\item \textbf{LQG-based scheme}~\cite{b11}: Following the LQR formulation introduced in Section~II, we compute the feedback gain $\bold K_n$ offline from \eqref{D1}-\eqref{D2}, and apply the control law $\bold u_n=-\bold K_n\hat{\bold e}_n$. Since this unconstrained input may violate the UAV's physical limits, we enforce feasibility via a two-step saturation/projection procedure: (i) clip the control input to satisfy $\|\bold u_n\|\le a_{\max}$; and (ii) if $\|\bold v_n+T\bold u_n\|>V_{\max}$, project the predicted velocity $\bold v_n+T\bold u_n$ onto the $\ell_2$ speed ball of radius $V_{\max}$ and then obtain the corresponding $\bold u_n$, ensuring $\|\bold v_n+T\bold u_n\|\le V_{\max}$. After determining the feasible input $\bold u_n$, the beamforming vector $\bold{w}_{n+1}$ is obtained according to (\ref{E-49}).
	
	\item \textbf{Non-casual MPC scheme}: At each time slot, this benchmark assumes that the target states over the next $N_0$ steps are perfectly known. As a result, the sensing constraint \eqref{P1-b} is removed and problem \eqref{P1} reduces to a deterministic convex MPC, which can be solved directly. Since it bypasses EKF-based state estimation by assuming future target states, it is used only as a clairvoyant trajectory reference, and its estimation performance is not reported.
\end{itemize}

	\begin{figure*}[t]
	\centering
	\begin{minipage}{0.315\textwidth}
		\centering
		\includegraphics[width=\linewidth]{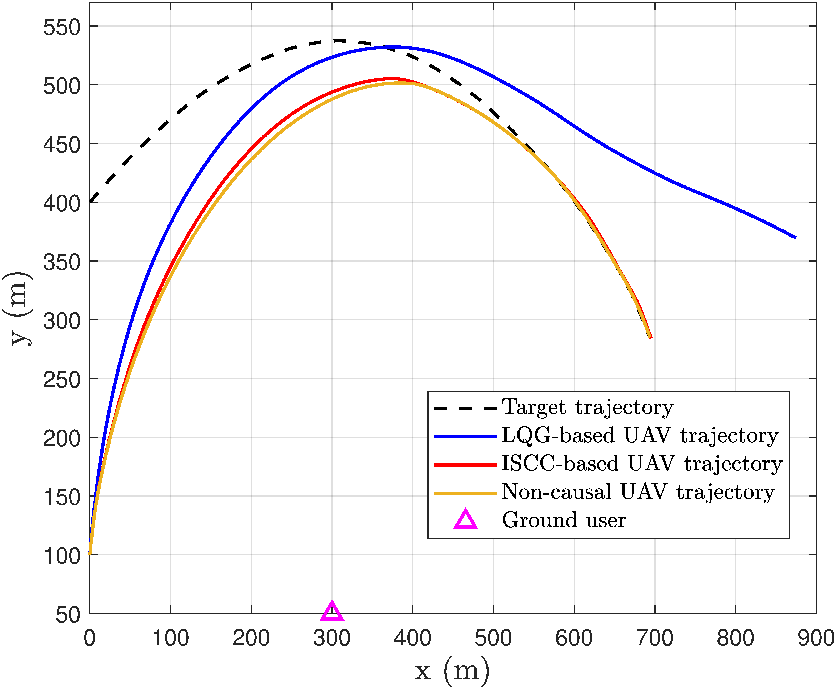}
		\subcaption{Case 1: $\bold{p}^{U}[1]=(0,100)$ m}
		\label{T100}
	\end{minipage}
	\hfill
	\begin{minipage}{0.315\textwidth}
		\centering
		\includegraphics[width=\linewidth]{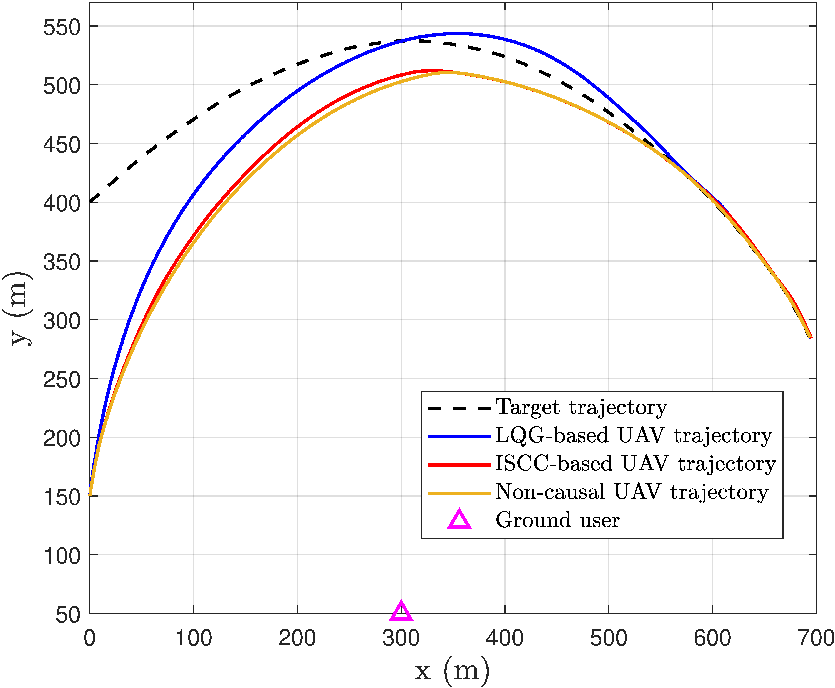}
		\subcaption{Case 2: $\bold{p}^{U}[1]=(0,150)$ m}
		\label{T150}
	\end{minipage}
	\hfill
	\begin{minipage}{0.315\textwidth}
		\centering
		\includegraphics[width=\linewidth]{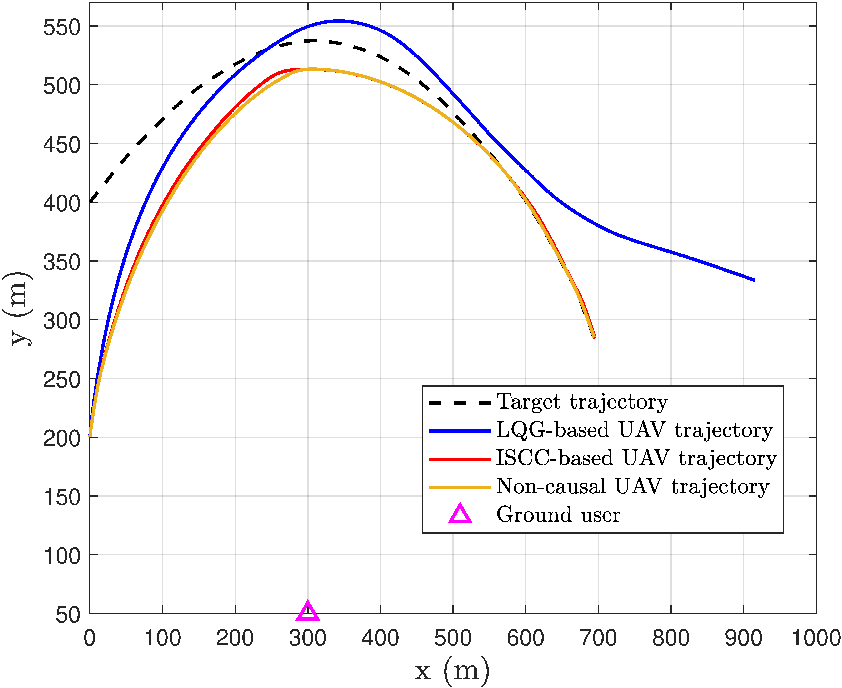}
		\subcaption{Case 3: $\bold{p}^{U}[1]=(0,200)$ m}
		\label{T200}
	\end{minipage}
	\caption{\raggedright UAV trajectories obtained by different schemes under an accurate initial target state estimate.}
	\label{fig4}
	
	\vspace{0pt}
\end{figure*}

\begin{figure*}[t]
	\centering
	\begin{minipage}{0.315\textwidth}
		\centering
		\includegraphics[width=\linewidth]{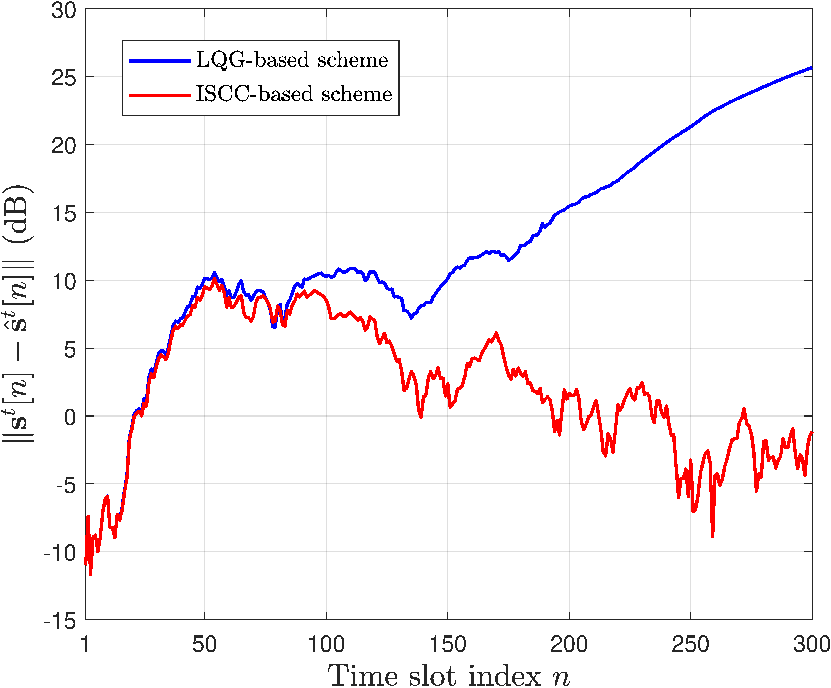}
		\subcaption{Case 1: $\bold{p}^{U}[1]=(0,100)$ m}
		\label{E100}
	\end{minipage}
	\hfill
	\begin{minipage}{0.315\textwidth}
		\centering
		\includegraphics[width=\linewidth]{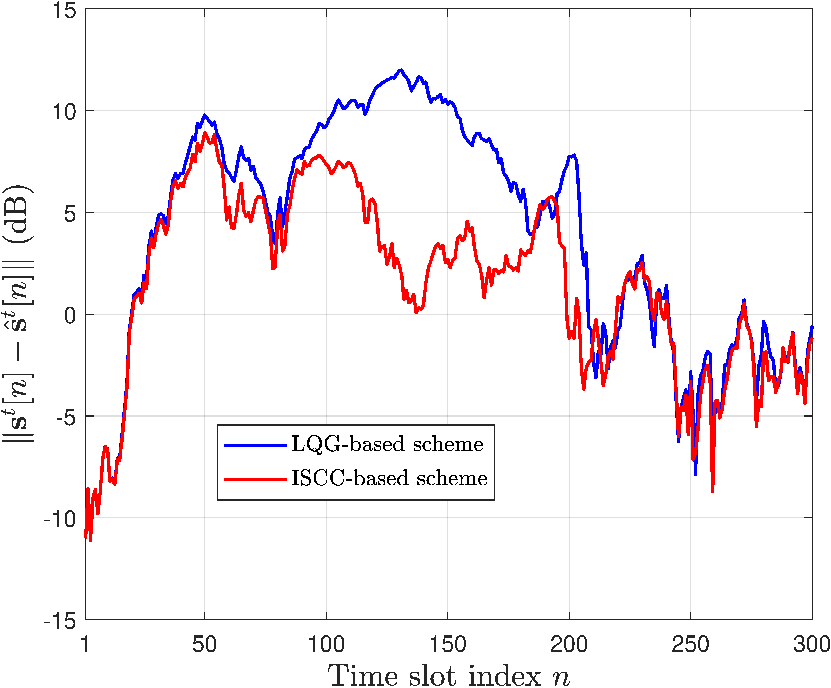}
		\subcaption{Case 2: $\bold{p}^{U}[1]=(0,150)$ m}
		\label{E150}
	\end{minipage}
	\hfill
	\begin{minipage}{0.315\textwidth}
		\centering
		\includegraphics[width=\linewidth]{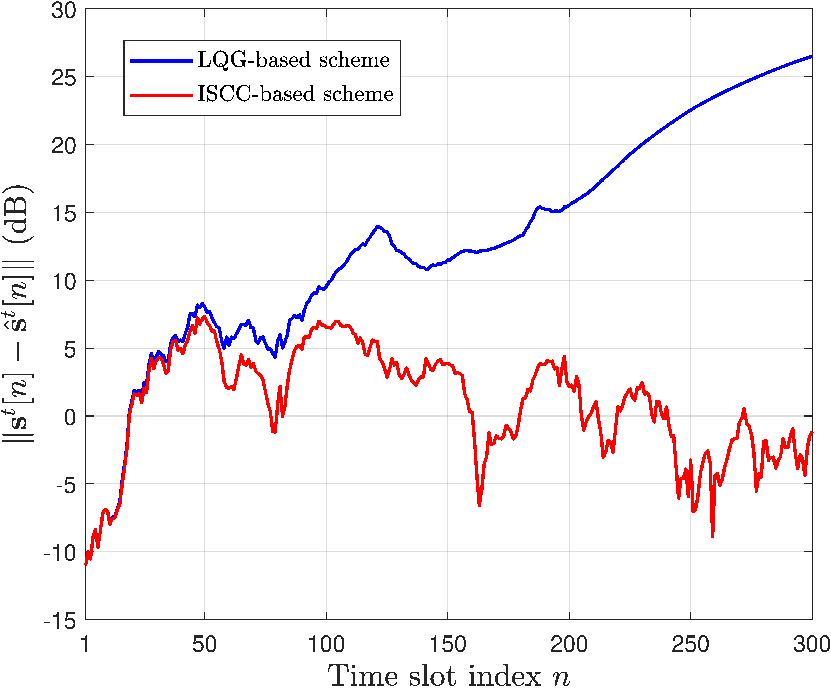}
		\subcaption{Case 3: $\bold{p}^{U}[1]=(0,200)$ m}
		\label{E200}
	\end{minipage}
	\caption{\raggedright Target estimation errors obtained by different schemes under an accurate initial target state estimate.}
	\label{fig5}
	
	\vspace{-8pt}
\end{figure*}

First, Fig.~\ref{fig4} illustrates the UAV trajectories obtained by the proposed ISCC scheme, the LQG-based benchmark, and the non-causal benchmark under an accurate initial estimate of the target state, characterized by $\hat{\mathbf{M}}_{1}=\mathbf{Q}_s$, for three different UAV initial positions, i.e., Case~1 with $\bold p^{U}[1]=(0,100)$~m, Case~2 with $\bold p^{U}[1]=(0,150)$~m, and Case~3 with $\bold p^{U}[1]=(0,200)$~m.
As shown, the proposed ISCC scheme closely follows the target trajectory in all three cases while satisfying the per-slot communication requirement. In contrast, the LQG-based benchmark exhibits a pronounced dependence on the initial position, i.e., it tracks well in Case~2 but gradually deviates in Cases~1 and~3, indicating a less stable behavior. This sensitivity stems from the coupled effect of estimation-based feedback and constraint enforcement.
Specifically, the controller first generates a candidate input $\bold u_n=-\bold K_n\hat{\bold e}_n$ and then applies saturation/projection to satisfy the physical limits. While this operation guarantees feasibility, it introduces a nonlinear mismatch relative to the nominal LQG feedback, and the extent of such mismatch depends on the evolving geometry and the control effort required from a given initial position. As a result, the LQG-based design becomes more sensitive to the operating condition, once the UAV enters a less favorable regime, the constrained control action may fail to adequately compensate for the tracking error, which can in turn degrade the sensing geometry and weaken subsequent EKF updates, leading to an error-accumulation effect and trajectory divergence over the long horizon. Moreover, it is observed that a visible gap between the ISCC and the non-causal trajectories appears in the early stage, while the two trajectories become nearly overlapped in the second half of the mission. This behavior is expected due to both causality and constraint differences. At the beginning of the mission, the UAV is far from the target and the EKF-based estimate/prediction is less reliable due to limited informative measurements; meanwhile, the ISCC scheme must explicitly satisfy the sensing requirement, which can significantly restrict the feasible maneuvering region when the sensing link is weak. Consequently, the ISCC trajectory tends to be more conservative than the non-causal benchmark. As the mission proceeds, the UAV approaches the target and accumulates more reliable measurements, which improves the EKF accuracy and relaxes the sensing constraint; therefore, the prediction mismatch diminishes and the ISCC trajectory progressively approaches the non-causal one, leading to an almost overlapping behavior in the later stage.

The above trajectory behaviors are further corroborated in Fig.~\ref{fig5}, which plots the target estimation error for the three cases. Note that the non-causal MPC benchmark is excluded since it is oracle-aided and does not involve online state estimation. For Cases~1 and~3, the LQG-based benchmark suffers from a steadily increasing estimation error, suggesting an error-accumulation effect, i.e.,  once the UAV starts to drift away from the target, the sensing geometry degrades, which weakens the EKF update and further amplifies the control deviation.
For Case~2, the LQG-based benchmark maintains a relatively small estimation error, which explains why it can achieve successful tracking in Fig.~\ref{fig4} (b). By contrast, the proposed ISCC scheme achieves consistently lower estimation errors across all cases, since it explicitly enforces sensing reliability (via the echo-SNR constraint) when coordinating beamforming and control decisions, thereby stabilizing the measurement quality for EKF updates.

	\begin{figure}[t]
	\centering
	\begin{minipage}{0.241\textwidth}
		\centering
		\includegraphics[width=\linewidth]{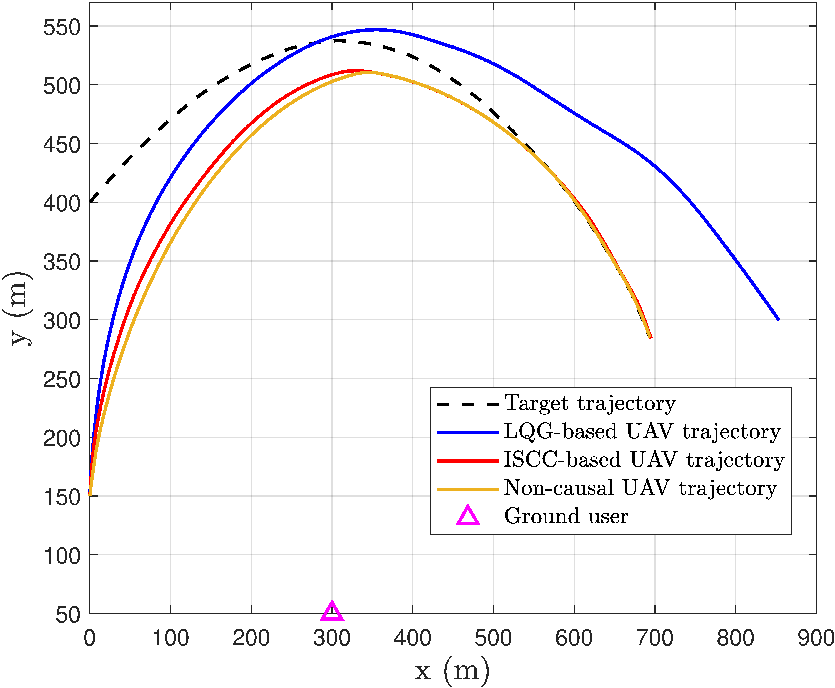}
		\subcaption{UAV trajectories}
		\label{T150_I}
	\end{minipage}
	\hfill
	\begin{minipage}{0.241\textwidth}
		\centering
		\includegraphics[width=\linewidth]{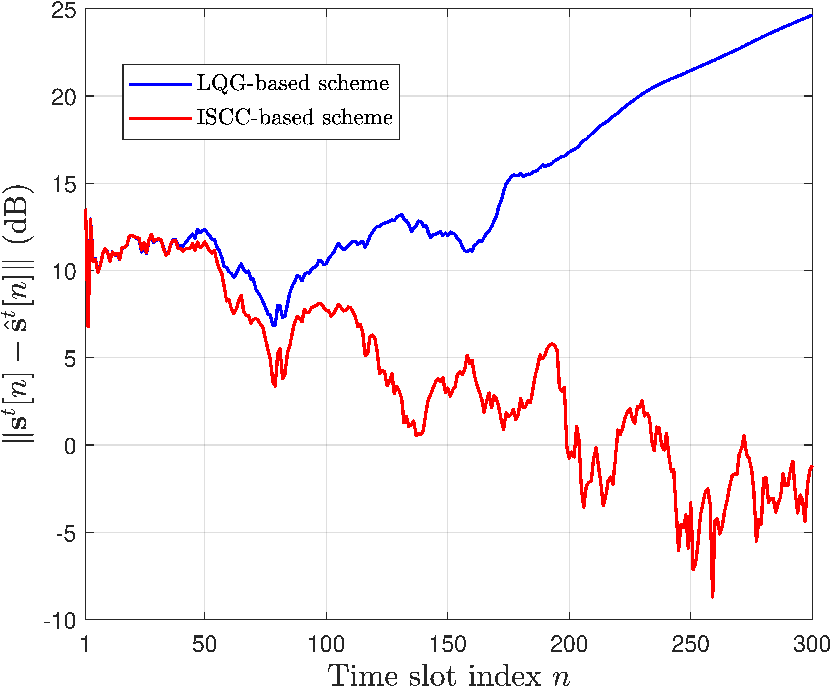}
		\subcaption{Target estimation errors}
		\label{R150_I}
	\end{minipage}
	\caption{\raggedright UAV trajectories and target estimation errors in Case 2 under an inaccurate initial target state estimate.}
	\label{fig7}
	
	\vspace{0pt}
\end{figure}

Fig.~\ref{fig7} further evaluates the performance of the considered schemes in Case 2 with an inaccurate initial estimate of the target state, characterized by $\hat{\mathbf{M}}_{1}=\bold{Q}_s+\text{diag}([400,400,4,4])$. As can be observed, compared with the accurate-initialization case in Fig.~\ref{fig4} (b), the LQG-based benchmark now exhibits a clear long-term deviation, accompanied by a rapidly growing target state estimation error. This behavior suggests that, under poor initialization, the LQG-based scheme may steer the UAV into trajectories that do not sustain sufficient sensing quality, causing the EKF error to persist and accumulate, which ultimately deteriorates tracking performance. In contrast, the proposed ISCC scheme remains stable, it drives the UAV along a trajectory that sustains effective sensing, thereby keeping the EKF error bounded throughout the horizon. These results highlight the superior performance of the proposed scheme even with inaccurate initial target state information.

Next, to further corroborate the above representative-instance observations under random measurement noise $\bold{z}_n$ and target motion process noise $\bold{n}_s$, we report Monte Carlo results (see Figs. \ref{fig8}--\ref{fig11}) in Case~2 with $M=100$ independent trials to quantify the long-horizon tracking and communication performance. Specifically, Fig.~\ref{fig8} evaluates the long-horizon tracking quality by plotting the root mean square (RMS) target state tracking error over time, calculated as $\sqrt{\frac{1}{M}\sum_{m=1}^M\Vert\bold{e}_n^{(m)}\Vert^2}$, where $\bold{e}_n^{(m)}$ denotes the target state tracking error at the $n$-th time slot for the $m$-th trial. As can be observed, for the LQG-based benchmark, the tracking error may decrease in the early stage but fails to remain small in the later stage under both accurate and inaccurate initializations, implying that its long-term closed-loop tracking behavior is not reliably maintained across random realizations. Moreover, inaccurate initialization leads to more pronounced late-stage degradation, since a larger initial mismatch makes it more likely for some trials to evolve into unfavorable closed-loop realizations with noticeable drift, which elevates the error curve. By contrast, the proposed ISCC scheme drives the tracking error to a consistently low level toward the end of the horizon under both initialization settings, since it performs receding-horizon control while explicitly accounting for sensing feasibility, thereby stabilizing the constrained closed-loop evolution. In addition, its performance closely approaches the oracle-aided non-causal MPC reference enabled by future target-state information.

\begin{figure}[t]
	\centering
	\includegraphics[width=0.9\linewidth]{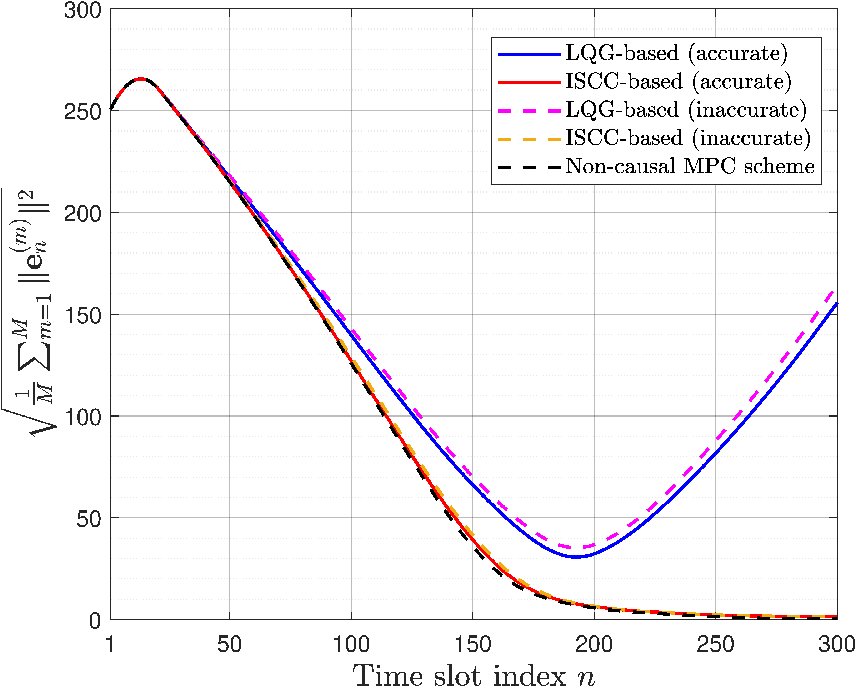}
	\caption{RMS target state tracking error in Case 2.}
	\label{fig8}
	\vspace{0pt}
\end{figure}

\begin{figure}[t]
	\centering
	\includegraphics[width=0.87\linewidth]{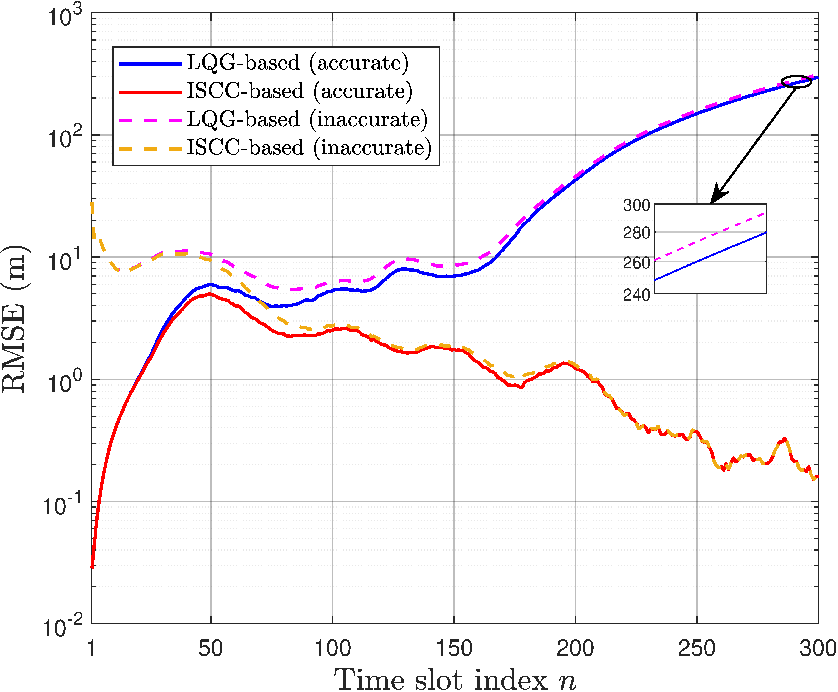}
	\caption{RMSE of target position estimation in Case 2.}
	\label{fig9}
	\vspace{0pt}
\end{figure}

\begin{figure}[t]
	\centering
	\includegraphics[width=0.87\linewidth]{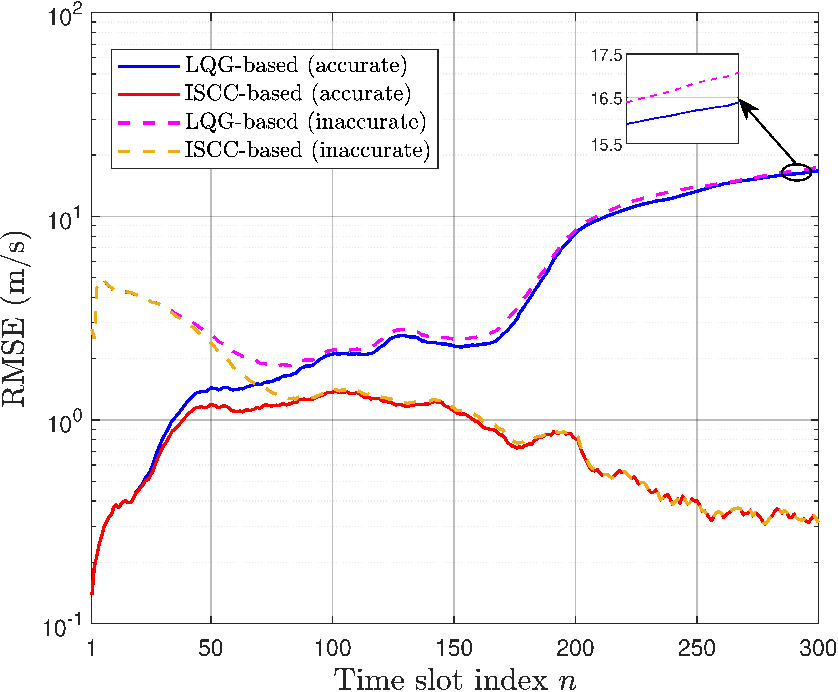}
	\caption{RMSE of target velocity estimation in Case 2.}
	\label{fig10}
	\vspace{0pt}
\end{figure}

\begin{figure}[t]
	\centering
	\includegraphics[width=0.9\linewidth]{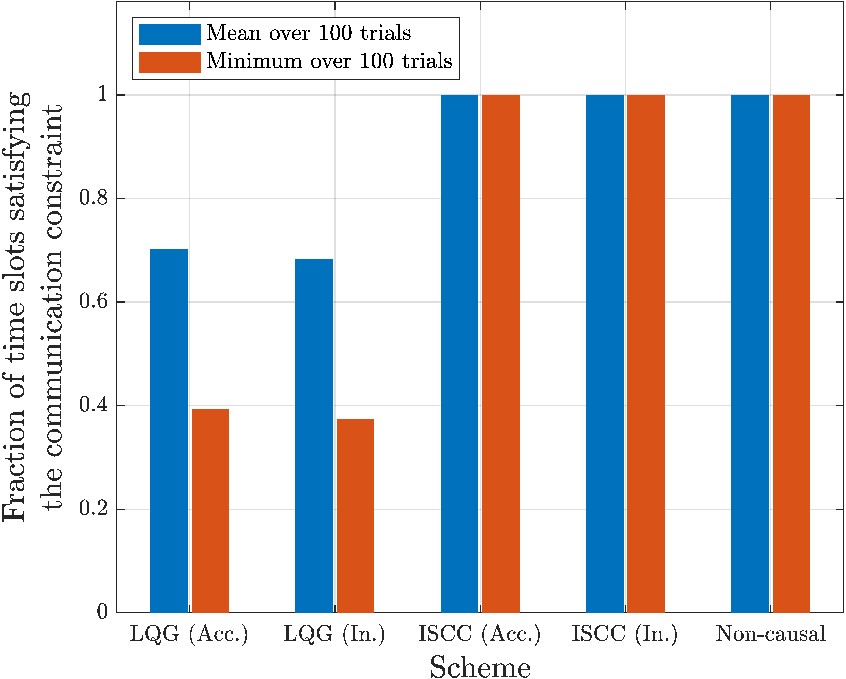}
	\caption{Fraction of time slots satisfying the communication rate constraint (\ref{P1-a}) in Case 2.}
	\label{fig11}
	\vspace{0pt}
\end{figure}

In Figs.~\ref{fig9} and~\ref{fig10}, we demonstrate the radar sensing performance in terms of root mean squared error (RMSE) for target position and velocity tracking, respectively. It is observed that the proposed ISCC scheme consistently achieves much lower RMSE than the LQG-based benchmark for both position and velocity, and the performance gap becomes increasingly pronounced as the mission proceeds. In particular, the LQG-based benchmark exhibits a clear RMSE growth in the later stage under both accurate and inaccurate initializations. This behavior is consistent with the observations in Figs.~\ref{fig4}--\ref{fig7}, i.e., while some trials may maintain acceptable estimation accuracy, once a subset of trials starts to lose track, the EKF errors can escalate rapidly and dominate the Monte Carlo statistics, thereby driving up the RMSE curves. With inaccurate initialization, such tracking-failure trials occur more frequently, leading to an additional RMSE increase compared with the accurate-initialization case. In contrast, the proposed ISCC scheme keeps the RMSE bounded and generally improves it over time, since it explicitly enforces sensing reliability when coordinating beamforming and control decisions, which sustains informative measurements for EKF updates and suppresses the occurrence of tracking-failure realizations.

\vspace{1pt}
Finally, Fig.~\ref{fig11} evaluates the communication reliability in terms of the fraction of time slots satisfying the per-slot rate constraint~(\ref{P1-a}). For each trial, this fraction is calculated over all $N$ time slots. The mean fraction corresponds to the average across $M$ trials, while the minimum fraction denotes the lowest fraction observed among these trials. As shown, both the proposed ISCC scheme and the non-causal benchmark achieve a satisfaction fraction of exactly $1$ for both the mean and the minimum fraction, indicating that the rate constraint is satisfied at every time slot in every trial. By contrast, the LQG-based benchmark attains a noticeably smaller mean fraction and a substantially smaller minimum fraction, which implies that rate violations can occur in a non-negligible portion of time slots under unfavorable realizations. Moreover, the satisfaction fraction under inaccurate initialization is lower than that under accurate initialization, since the larger initial mismatch increases the likelihood of adverse realizations where tracking/estimation quality deteriorates and the resulting closed-loop evolution steers the UAV to less favorable communication geometries, thereby increasing the frequency of rate-constraint violations.

\section{Conclusion}\label{sec:VI}
This paper proposed an ISCC framework for UAV-assisted mobile target tracking, where the UAV must simultaneously ensure tracking accuracy, reliable communication, and physically feasible motion. By modeling the tracking process as a discrete-time control system, we formulated a stochastic MPC problem for joint control and beamforming design. To make the MPC problem tractable, we first employed an EKF to estimate the target state, and then derived the closed-form optimal beamforming solution for any given control input, which enables an equivalent control-oriented reformulation. An asymptotically tight lower-bound relaxation technique was further developed to obtain a convex problem that can be efficiently solved. Numerical results demonstrated that the proposed ISCC approach achieves tracking performance comparable to a non-causal benchmark while maintaining stable communication service, and it consistently outperforms a conventional LQG-based scheme under practical constraints.

\begin{appendices}
	\section{}
To solve problem (\ref{C-check}), we decompose it into $N_0$ independent subproblems, each corresponding to a single time slot. Specifically, the $i$-th subproblem is formulated as 
\vspace{-5pt}

\begin{subequations}\label{C-sub}\small
	\begin{alignat}{2}
		\underset{\bold{w}_{n+i}}{\text{max}} \ & \log_2\left(1 + \frac{\beta_0 |\bold{a}_{c,i}^H \bold{w}_{n+i}|^2}{\sigma_c^2d_{c,i}^2 } \right) - R_{th} \notag \\
		\text{s.t.} \ & |\check{\bold{a}}_{n+i}^H \bold{w}_{n+i}|^2 \ge \mathbb{E}_{\overline{\bold{n}}_i}\{d_{n+i}^4\} \Gamma_{th}, \label{C-sub-a}\\
		& \Vert \bold{w}_{n+i} \Vert^2 \le P_T. \label{C-sub-b}
	\end{alignat}
\end{subequations}
Despite the non-convexity of this problem, its optimal beamforming vector admits a closed-form expression given by\cite{I12}
\begin{equation}\label{C-w-opt}\small
	\bold{w}^*_{n+i} =
	\begin{cases}
		\sqrt{P_T} \frac{\bold{a}_{c,i}}{\Vert \bold{a}_{c,i} \Vert}, &\text{if}\  P_T|\check{\bold{a}}_{n+i}^H \bold{a}_{c,i}|^2 \ge M_t\Gamma_{d,i},\\
		\epsilon_1 \bm{\chi}_1 + \epsilon_2 \bm{\chi}_2, & \text{otherwise},
	\end{cases}
\end{equation}
where $\Gamma_{d,i}\triangleq\mathbb{E}_{\overline{\bold{n}}_i}[d_{n+i}^4] \Gamma_{th}$, $\epsilon_1,\epsilon_2,\bm{\chi}_1$ and $\bm{\chi}_2$ are auxiliary terms expressed as
\begin{equation}\small
	\epsilon_1=\sqrt{\frac{\Gamma_{d,i}}{M_t}}\frac{\bm{\chi}^H_1\bold{a}_{c,i}}{|\bm{\chi}^H_1\bold{a}_{c,i}|},\epsilon_2=\sqrt{P_T-\frac{\Gamma_{d,i}}{M_t}}\frac{\bm{\chi}^H_2\bold{a}_{c,i}}{|\bm{\chi}^H_2\bold{a}_{c,i}|}.\notag
\end{equation}
\begin{equation}\small
	\bm{\chi}_1=\frac{\check{\bold{a}}_{n+i}}{\Vert\check{\bold{a}}_{n+i}\Vert},\bm{\chi}_2=\frac{\bold{a}_{c,i}-(\bm{\chi}^H_1\bold{a}_{c,i})\bm{\chi}_1}{\Vert\bold{a}_{c,i}-(\bm{\chi}^H_1\bold{a}_{c,i})\bm{\chi}_1\Vert},\notag
\end{equation}

By substituting $\bold{w}^*_{n+i}$ from (\ref{C-w-opt}) back into problem (\ref{C-sub}), the corresponding optimal achievable value of the beamforming-dependent term $G_{n+i}^*\triangleq\vert\bold{a}_{c,i}\bold{w}_{n+i}^*\vert^2$ is
\begin{equation}\label{C-R-opt}\small
	G^*_{n+i} =
	\begin{cases}
		\gamma, \ \ \text{if} \ \gamma\cos^2\theta_{i} \geq \Gamma_{d,i}, \\
		\left(\sqrt{\Gamma_{d,i}}\cos\theta_{i}+\sqrt{\gamma-\Gamma_{d,i}}\sin\theta_{i}\right)^2,\text{otherwise}.\notag
	\end{cases}
\end{equation}
Therefore, the optimal objective value of problem (\ref{C-check}) is $R^* = \min_{i\in\mathcal{I}} \left\{ R_{n+i}^* - R_{th} \right\}$, where $R_{n+i}^*\triangleq\log_2\left(1 + \frac{\beta_0 G_{n+i}^*}{\sigma_c^2d_{c,i}^2}\right)$ denotes the maximum achievable communication rate at the $(n+i)$-th time slot under the given control inputs $\{\bold{u}_{n-1+i}\}_{i\in\mathcal{I}}$.

\section{Proof of Lemma \ref{lemma-1}}
	To establish (\ref{E-55}), it suffices to prove that, for each $i\in\mathcal I$,
	\begin{equation}\label{E-77}\small
		\Gamma^*_{n+i}\ge \Gamma_{d,i} \Longleftrightarrow R_{n+i}^*\ge R_{th}.\notag
	\end{equation}
	To this end, we define the following feasible sets:
	\begin{subequations}\small
	\begin{align}
	\mathcal S_i \triangleq \{\mathbf w\in\mathbb C^{M_t}:\ \|\mathbf w\|^2\le P_T,\ R_{n+i}(\mathbf w)\ge R_{th}\},\\
	\mathcal T_i \triangleq \{\mathbf w\in\mathbb C^{M_t}:\ \|\mathbf w\|^2\le P_T,\ |\check{\mathbf a}_{n+i}^H\mathbf w|^2\ge \Gamma_{d,i}\},
	\end{align}
\end{subequations}
where $\mathcal S_i$ and $\mathcal T_i$ correspond to the feasibility regions of problems (\ref{S-sub}) and (\ref{C-sub}), respectively. With these definitions, we can write $\Gamma^*_{n+i}=\max_{\mathbf w\in\mathcal S_i} |\check{\mathbf a}_{n+i}^H\mathbf w|^2$ and $R^*_{n+i}=\max_{\mathbf w\in\mathcal T_i} R_{n+i}(\mathbf w)$. 

\textit{1) Proof of $\Gamma^*_{n+i}\ge \Gamma_{d,i}\Rightarrow R^*_{n+i}\ge R_{th}$.} Let $\mathbf{w}_s\in\mathcal S_i$ denote an optimal solution of problem (\ref{S-sub}). By definition, we have \(\Gamma^*_{n+i}=|\check{\mathbf a}_{n+i}^H\mathbf{w}_s|^2\). Under the assumption $\Gamma^*_{n+i}\ge \Gamma_{d,i}$, it follows that
$|\check{\mathbf a}_{n+i}^H\mathbf w_s|^2\ge \Gamma_{d,i}$, i.e., $\mathbf w_s\in\mathcal T_i$. Thus, \(\mathbf{w}_s\in\mathcal S_i\cap\mathcal T_i\) and $R^*_{n+i}=\max_{\mathbf w\in\mathcal T_i} R_{n+i}(\mathbf w)\ge R_{n+i}(\mathbf w_s)\ge R_{th}$, where the last inequality holds since $\mathbf w_s\in\mathcal S_i$.

\textit{2) Proof of $R^*_{n+i}\ge R_{th}\Rightarrow \Gamma^*_{n+i}\ge \Gamma_{d,i}$.} Similarly, let $\mathbf{w}_{\tau}\in\mathcal T_i$ denote an optimal solution of problem (\ref{C-sub}). Then, we have \(R^*_{n+i}=R_{n+i}(\mathbf w_{\tau})\). Under the assumption $R^*_{n+i}\ge R_{th}$, it follows that
$R_{n+i}(\mathbf w_{\tau})\ge R_{th}$, i.e., $\mathbf w_{\tau}\in\mathcal S_i$. Hence, \(\mathbf{w}_{\tau}\in\mathcal S_i\cap\mathcal T_i\) and $\Gamma^*_{n+i}=\max_{\mathbf w\in\mathcal S_i} |\check{\mathbf a}_{n+i}^H\mathbf w|^2\ge |\check{\mathbf a}_{n+i}^H\mathbf w_{\tau}|^2\ge \Gamma_{d,i}$, where the last inequality holds since $\mathbf w_{\tau}\in\mathcal T_i$.

Combining the two implications completes the proof.

\section{Proof of Lemma \ref{lemma-2}}
To establish the lower bound in (\ref{lower}), we examine the following two cases based on the value of $\delta_i$.

\textbf{1) Case I: $\delta_i = 0$}.

In this case, $\check{\bold{p}}^{t}[n+i]=\bold{p}^{u}$ holds, which directly yields $\check{\bold{a}}_{n+i}=\bold{a}_{c,i}$ and $\theta_i=0$.  
Substituting $\theta_i=0$ into (\ref{Optimal}) gives $\Gamma^*_{n+i}=\gamma$.  
Therefore, (\ref{lower}) reduces to $\Gamma^*_{n+i}=\Gamma^l_{n+i}=\gamma$, showing that the lower bound is exact (tight) in this case.

\textbf{2) Case II: $\delta_i = 1$}.

We first show that $\Gamma_{n+i}^l$ is a valid lower bound. As can be observed, it is sufficient to prove that
\begin{equation}\label{prove}\small
	\left(\sqrt{\eta d^2_{c,i}}\cos\theta_{i}+\sqrt{\gamma-\eta d^2_{c,i}}\sin\theta_{i}\right)^2 \ge \gamma-\eta d^2_{c,i},
\end{equation}
holds under the condition $\gamma\cos^2\theta_{i} < \eta d^2_{c,i}$. Let $\phi_i = \arccos\sqrt{\eta d_{c,i}^2/\gamma} \in \left[0,\frac{\pi}{2}\right]$, and by substituting it into (\ref{prove}), we obtain
\begin{subequations}\small
	\begin{alignat}{2}
		&\left(\sqrt{\eta d^2_{c,i}}\cos\theta_{i}+\sqrt{\gamma-\eta d^2_{c,i}}\sin\theta_{i}\right)^2-(\gamma-\eta d^2_{c,i})\notag\\
		&\qquad=\gamma\left(\cos\phi_i\cos\theta_{i}+\sin\phi_i\sin\theta_{i}\right)^2-\gamma\sin^2\phi_i\notag\\
		&\qquad= \gamma\cos^2(\theta_{i}-\phi_i) - \gamma\cos^2\left(\tfrac{\pi}{2}-\phi_i\right) \triangleq d_{if}. \notag
	\end{alignat}
\end{subequations}
Since $\gamma\cos^2\theta_i \le \eta d_{c,i}^2 = \gamma\cos^2\phi_i$, we have $0 \le \theta_i - \phi_i \le \frac{\pi}{2}-\phi_i$, which implies $\cos^2(\theta_i - \phi_i) \ge \cos^2(\tfrac{\pi}{2}-\phi_i)$ and hence $d_{if}\ge 0$. This confirms that $\Gamma^*_{n+i} \ge \Gamma_{n+i}^l$ for the case of $\delta_i=1$.

Next, we establish the asymptotic tightness of the bound as $M_t\to\infty$. For notational convenience, let $\Phi_{u,i}\triangleq\Phi(\bold{p}^{U}[n+i],\bold{p}^{u})$, $\Omega_{u,i}\triangleq\Omega(\bold{p}^{U}[n+i],\bold{p}^{u})$, $\check{\Phi}_{t,i}\triangleq\Phi(\bold{p}^{U}[n+i],\check{\bold{p}}^{t}[n+i])$ and $\check{\Omega}_{t,i}\triangleq\Omega(\bold{p}^{U}[n+i],\check{\bold{p}}^{t}[n+i])$. When $\Phi_{u,i}\neq \check{\Phi}_{t,i}$ and $\Omega_{u,i}\neq \check{\Omega}_{t,i}$, the inner product $\lvert\check{\bold{a}}_{n+i}^H\bold{a}_{c,i}\rvert$ satisfies
\begin{subequations}\small
	\begin{alignat}{2}
		\vert\check{\bold{a}}_{n+i}^H\bold{a}_{c,i}\vert
		&= \left| 
		\sum_{m_x=0}^{M_x^t-1} e^{j\pi m_x (\Phi_{u,i}-\check\Phi_{t,i})}
		\sum_{m_y=0}^{M_y^t-1} e^{j\pi m_y (\Omega_{u,i}-\check\Omega_{t,i})}
		\right| \notag\\
		&= 
		\left|
		\left(
		\frac{e^{j\pi M_x^t\Delta_{\Phi,i}} - 1}
		{e^{j\pi\Delta_{\Phi,i}} - 1}
		\right)	\left(
		\frac{e^{j\pi M_y^t\Delta_{\Omega,i}} - 1}
		{e^{j\pi\Delta_{\Omega,i}} - 1}
		\right)
		\right| \notag\\
		&=
		\Bigg|
		\frac{\sin\!\left(\pi M_x^t \Delta_{\Phi,i}/2\right)}
		{\sin\!\left(\pi \Delta_{\Phi,i}/2\right)}\Bigg|
		\Bigg|
		\frac{\sin\!\left(\pi M_y^t \Delta_{\Omega,i}/2\right)}
		{\sin\!\left(\pi \Delta_{\Omega,i}/2\right)}\Bigg|,\notag
	\end{alignat}
\end{subequations}
where $\Delta_{\Phi,i}\triangleq \Phi_{u,i}-\check{\Phi}_{t,i}$ and $\Delta_{\Omega,i}\triangleq \Omega_{u,i}-\check{\Omega}_{t,i}$. Based on the definition $\theta_{i}\triangleq\arccos\frac{\vert\check{\bold{a}}_{n+i}^H\bold{a}_{c,i}\vert}{\Vert\check{\bold{a}}_{n+i}\Vert\Vert\bold{a}_{c,i}\Vert}\in \left[0,\frac{\pi}{2}\right]$, we have
$\cos\theta_i = \Bigg|
	\frac{\sin\!\left(\pi M_x^t \Delta_{\Phi,i}/2\right)}
	{M_x^t\sin\!\left(\pi \Delta_{\Phi,i}/2\right)}\Bigg|
	\Bigg|
	\frac{\sin\!\left(\pi M_y^t \Delta_{\Omega,i}/2\right)}
	{M_y^t\sin\!\left(\pi \Delta_{\Omega,i}/2\right)}\Bigg|$.
When either $M_x^t \to \infty$ or $M_y^t \to \infty$, we obtain $\cos\theta_i \to 0$, meaning $\theta_i \to \frac{\pi}{2}$ and consequently $d_{if}\to 0$.  
Similarly, in the partial-alignment cases $\Delta_{\Phi,i}=0$ or $\Delta_{\Omega,i}=0$, $\theta_i\to\frac{\pi}{2}$ still holds as the remaining array dimension grows. Since $\delta_i=1$ excludes the perfect-alignment condition $\Delta_{\Phi,i}=\Delta_{\Omega,i}=0$, the limit $\theta_i\to\frac{\pi}{2}$ is attainable as $M_t\to\infty$. Therefore, the lower bound $\Gamma^l_{n+i}$ becomes asymptotically tight as $M_t\to\infty$.

\section{Proof of Theorem \ref{theorem-1}}
To establish the convexity of problem (\ref{P5}), it suffices to show that $\bm{\Upsilon}\succeq 0$, $\bm{\Xi}_i\succeq 0$ for all $i\in\mathcal{I}$, and that the quartic function $
	h_i(\hat{\mathbf u}_n)
	=\Big((\mathbf S_i\hat{\mathbf u}_n+\overline{\mathbf c}_i)^T\bm\Lambda_i^c(\mathbf S_i\hat{\mathbf u}_n+\overline{\mathbf c}_i)\Big)^2$ is convex for any \(i\in\mathcal I\). 
First, it can be readily seen that  
\begin{equation}\small
	\bm{\Upsilon} \triangleq \sum_{i=1}^{N_0} \mathbf{S}_{i}^{T} \bm{\Lambda}_{i}^{q} \mathbf{S}_{i} + \mathbf{I}_{N_0} \otimes \mathbf{R} \succeq 0,\notag
\end{equation}  
since \(\bm{\Lambda}_{i}^q \succeq 0\), \(\overline{\mathbf N}_i \succeq 0\), and \(\mathbf{R} \succ 0\). Furthermore, as \(\eta\) and \(\Gamma_{th}\) are strictly positive to ensure the meaningfulness of the associated terms, it follows from the definition in (\ref{E-70}) that \(\bm{\Xi}_i \succeq 0\) for all \(i \in \mathcal{I}\). Next, to prove the convexity of \(h_i(\hat{\mathbf u}_n)\), we define the following quadratic function:
\begin{equation}\small
q_i(\hat{\mathbf u}_n) \triangleq (\mathbf S_i \hat{\mathbf u}_n + \overline{\mathbf c}_i)^T \bm\Lambda_i^c (\mathbf S_i \hat{\mathbf u}_n + \overline{\mathbf c}_i),\notag
\end{equation}  
which is nonnegative since \(\bm{\Lambda}_i^c \succeq 0\). Its gradient and Hessian can be obtained as  
\begin{equation}\label{E-81}\small
\nabla q_i(\hat{\mathbf u}_n) = 2 \mathbf S_i^T \bm\Lambda_i^c (\mathbf S_i \hat{\mathbf u}_n + \overline{\mathbf c}_i), 
\nabla^2 q_i(\hat{\mathbf u}_n) = 2 \mathbf S_i^T \bm\Lambda_i^c \mathbf S_i.\notag
\end{equation} 
Using the chain and product rules, the Hessian of \(h_i(\hat{\mathbf u}_n) = \big(q_i(\hat{\mathbf u}_n)\big)^2\) can be computed as  
\begin{equation}\small
\nabla^2 h_i(\hat{\mathbf u}_n) = 2 \, \nabla q_i(\hat{\mathbf u}_n) \nabla q_i(\hat{\mathbf u}_n)^T + 2 \, q_i(\hat{\mathbf u}_n) \, \nabla^2 q_i(\hat{\mathbf u}_n).\notag
\end{equation}  
Since \(\nabla q_i(\hat{\mathbf u}_n) \nabla q_i(\hat{\mathbf u}_n)^T \succeq 0\) and \(\nabla^2 q_i(\hat{\mathbf u}_n) \succeq 0\), it follows that \(\nabla^2 h_i(\hat{\mathbf u}_n) \succeq 0\), which implies that \(h_i(\hat{\mathbf u}_n)\) is convex for any \(i \in \mathcal{I}\). This completes the proof.

\end{appendices}

\end{document}